\documentclass[10pt]{paper}
\usepackage[T1]{fontenc}
\usepackage[utf8]{inputenc}
\usepackage{geometry}
\usepackage{float}
\usepackage{bm}
\usepackage{amsthm}
\usepackage{amsmath}
\usepackage{amssymb}
\usepackage{graphicx}
\usepackage[authoryear]{natbib}
\usepackage{hyperref}

\makeatletter

\providecommand{\tabularnewline}{\\}

\theoremstyle{plain}
\newtheorem{thm}{\protect\theoremname}
  \theoremstyle{plain}
  \newtheorem{lem}[thm]{\protect\lemmaname}
  \theoremstyle{remark}
  \newtheorem{rem}[thm]{\protect\remarkname}
  \theoremstyle{plain}
  \newtheorem{prop}[thm]{\protect\propositionname}

\usepackage{amsthm}
\usepackage{amsfonts}
\usepackage{pifont}
\usepackage{multirow}
\usepackage{array}
\usepackage{bm}
\usepackage{dsfont}
\usepackage{lineno}
\usepackage{hyperref}

\makeatother

  \providecommand{\lemmaname}{Lemma}
  \providecommand{\propositionname}{Proposition}
  \providecommand{\remarkname}{Remark}
\providecommand{\theoremname}{Theorem}

\begin{document}

\title{Theory of Performance Participation Strategies}

\author{Julia Kraus%
\thanks{Chair of Mathematical Finance, Technische Universitaet Muenchen, Germany,
julia.kraus@tum.de%
}, Philippe Bertrand%
\thanks{GREQAM, University of Aix-Marseille II, France and Euromed Management,
Marseille, France%
}, Rudi Zagst%
\thanks{Chair of Mathematical Finance, Technische Universitaet Muenchen, Germany%
}}

\date{February 7, 2011}
\maketitle
\begin{abstract}
The purpose of this article is to introduce, analyze and compare two
performance participation methods based on a portfolio consisting
of two risky assets: Option-Based Performance Participation (OBPP)
and Constant Proportion Performance Participation (CPPP). By generalizing
the provided guarantee to a participation in the performance of a
second risky underlying, the new strategies allow to cope with well-known
problems associated with standard portfolio insurance methods, like
e.g. the CPPI cash lock-in. This is especially an issue in times of
market crisis. However, the minimum guaranteed portfolio value at
the end of the investment horizon is not deterministic anymore, but
subject to systematic risk instead. With respect to the comparison
of the two strategies, various criteria are applied such as comparison
of terminal payoffs and payoff distributions. General analytical expressions
for all moments of both performance participation strategies as well
as standard OBPI and CPPI are derived. Furthermore, dynamic hedging
properties are examined, in particular classical delta hedging. \end{abstract}

\section{Introduction\label{sec:Introduction}}

In this paper we introduce and analyze the class of performance participation
strategies. With this respect we define performance participation
strategies as financial strategies which are designed to provide a
minimum performance in terms of a fraction of the outcome of one risky
asset while keeping the potential for profits resulting from the outperformance
of another risky asset. Due to this minimum performance feature they
can be considered as a generalization of the well-known class of portfolio
insurance strategies.

While the provided guarantee is not deterministic anymore but subject
to systematic risk instead, these strategies avoid the cash lock-in
feature that face standard CPPI methods and thus are able to take
advantage of a possible market recovery. After a sharp market drop,
like e.g. at the beginning of 2009, the entire risk budget is maybe
exhausted and the portfolio fully invested in the cash market afterwards.
The defensive portfolio allocation then remains unchanged until the
end of the investment horizon (or the next reallocation date) and
prohibits to participate in a potential market regeneration. Consequently,
the CPPI portfolio will only return the riskless interest rate and
the associated costs of insurance significantly diminish the resulting
return.

To cope with the above issues we substitute the primary risk-free
asset with a second risky investment alternative, also called the
reserve asset. This allows to provide even in critical market situations,
where standard portfolio insurance approaches tend to fail, a participation
in the performance of a risky investment opportunity. In order to
minimize the additionally introduced risk one could e.g. think about
the minimum variance portfolio as a risky reserve asset, but also
riskier alternatives are possible.

In this paper we introduce two different performance participation
strategies, one static, option-based approach as well as a dynamic
portfolio reallocation rule. With respect to the former we pick up
the among practitioners very popular Best of Two%
\footnote{Note that the name 'Best of Two' is registered by the Conrad Hinrich
Donner Private Bank (see \citet{VittLeifeld2005}).%
} (Bo2) strategy. It was first introduced by \citet{DichtlSchlenger2002}
and mainly relies on the concept of so-called exchange options. An
exchange option written on a pair of risky assets $S_{1}$ and $S_{2}$
gives the option holder the right to exchange at maturity $T$ the
performance of one asset against the other.%
\footnote{See, e.g., \citet{Margrabe1978} for details.%
} Thus, by setting up a static portfolio consisting of an adequate
number of shares of one of the risky assets and the same number of
exchange options written on the second risky asset the investor will
receive at the end of the investment period the return (except for
strategy costs) of the better performing asset during the observation
horizon. In this way, by guaranteeing a performance participation
in one of the risky assets serious portfolio losses can be narrowed,
while keeping the potential of full participation in rising markets.
The based-upon OBPP (Option-Based Performance Participation) constitutes
a generalization of the Best of Two concept to provide general investor-defined
levels of performance participation. A similar approach was already
mentioned in \citet{Lindset2004} within the context of relative guarantees
for life insurance contracts or pension plans.

With respect to the latter dynamic approach we rely on the for portfolio
insurance purposes well-established CPPI concept. In their seminal
papers \citet{BlackJones1987} as well as \citet{BlackPerold1992} originally introduced the CPPI
approach on a portfolio consisting of two risky assets, i.e. with
stochastic floor. Nevertheless, in a wide range of the literature
in the field of portfolio insurance strategies the CPPI investment
rule is restricted to a constant, deterministic interest rate and
one risky asset. In this paper we pick up \citet{BlackPerold1992}'s
original idea to define the CPPP (Constant Proportion Performance
Participation) strategy as a dynamic approach to performance participation.
In analogy to the CPPI concept, the resulting strategy not only guarantees
a minimum performance participation in one of the risky assets but
also allows for a leveraged participation in the outperformance of
a second asset.

Within the scope of this paper we provide a detailed analysis and
comparison of the OBPP and the CPPP with respect to various criteria.
Although the two strategies were already mentioned in different areas
of the financial literature, to the authors' knowledge no profound
theoretical analysis was conducted so far. In the case of the OBPP
strategy the literature is scarce: Except for \citet{Margrabe1978}'s
basic paper about the evaluation of exchange options, there only exist
some empirical performance reviews with a focus on the practical application
of the Bo2 strategy, like e.g. the works of \citet{DichtlSchlenger2002,DichtlSchlenger2003}
and \citet{VittLeifeld2005} and more popular articles in practicioners'
journals. With respect to the CPPP strategy, as mentioned earlier,
the basic literature like e.g. \citet{BlackJones1987}, \citep{BlackRouhani1989}
or \citep{BertrandPrigent2005}, mainly restricts to the one-dimensional
case with one risky asset and a risk-free interest rate.

We therefore first of all provide a formalized and unified definition
of the two performance participation strategies. This enables us to
establish a very important relationship between standard portfolio
insurance and more general performance participation strategies. Based
on that finding, generalized analytical expressions for all moments
of the payoff distributions of the standard portfolio insurance strategies
as well as the built-upon performance participation strategies are
derived. The subsequent analysis is conducted in the spirit of \citep{BlackRouhani1989}
and \citet{BertrandPrigent2005} for portfolio insurance strategies.

The remainder of this paper is organized as follows: In Section 2,
we briefly introduce and discuss the two performance participation
strategies under consideration. We examine their final payoffs and
show that the newly introduced strategies can be directly linked to
the standard CPPI and the standard OBPI method. A detailed analytical
analysis of the moments of the resulting payoff distributions is conducted
in Section 3. With regard to the practical implementation Section
4 especially covers the dynamic behavior of the two strategies. To
conclude the analysis, Section 5 summarizes the main findings and
gives some concluding remarks.

\section{Definition of the OBPP and the CPPP Strategy}

\subsection{Financial market setup\label{sec:financialMarket}}

With respect to the theoretical analysis of the two performance participation
strategies we define a two-dimensional Black-Scholes model on the
filtered probability space $\left(\Omega,\mathcal{F},\mathbb{P},\mathbb{F}\right)$.
The financial market thus offers three investment possibilities: two
risky assets $S_{1}$, $S_{2}$ and a riskless cash account $S_{0}$
that are traded continuously in time during the investment period
$\left[0,T\right]$. Within the context of performance participation
strategies the time horizon $T$ is regarded as the time horizon for
the provided participation. The risk-free asset grows with constant
continuous interest rate $r$, i.e. $S_{0}(t)=e^{r\cdot t}$. The
evolution of the remaining two assets, such as a stock, stock portfolio
or market index, is subject to systematic risk and the corresponding
price process $S_{i}(t),\ 0\leq t\leq T$ of stock $i=1,2$ is modeled
by the geometric Brownian motion 
\begin{equation}
dS_{i}(t)=S_{i}(t)\left(\mu_{i}dt+\sum_{j=1}^{2}\sigma_{ij}dW_{j}(t)\right),\ S_{i}(0)=s_{i}>0.\label{eq:SDE_risky}
\end{equation}
Here, $\textbf{W}(t)=\left(W_{1}(t),W_{2}(t)\right)',\ 0\leq t\leq T$
denotes a standard two-dimensional Brownian motion with respect to
the real-world measure $\mathbb{P}$ and the Brownian filtration $\mathbb{F}=\left\{ \mathcal{F}_{t},0\leq t\leq T\right\} $.
The constant matrices $\mathbf{\mu}=\left(\mu_{1},\mu_{2}\right)'$
and $\bm{\sigma}=\left(\sigma_{i,j}\right)_{i,j=1,2}$ with 
\[
\bm{\sigma}=\left(\begin{array}{cc}
\sigma_{1} & 0\\
\rho_{12}\cdot\sigma_{2} & \sqrt{1-\rho_{12}^{2}}\cdot\sigma_{2}
\end{array}\right),
\]
describe the drifts, the volatilities and the correlations of the
asset prices, where we assume $\mu_{2}\geq\mu_{1}\geq r\geq0$ and
$\sigma_{2}\geq\sigma_{1}>0$. Due to these risk-return characteristics
now and in the following we will call asset $S_{1}$ the reserve asset
and the riskier asset $S_{2}$ the active asset.%
\footnote{Note that this notation was already used in the early papers of \citet{BlackJones1987}
and \citet{BlackPerold1992}.%
} Furthermore, in order to eliminate any arbitrage opportunities the
matrix $\bm{\sigma}$ has to be regular inducing $\rho_{12}\in\left(-1,1\right)$.
The two risky underlyings are thus not perfectly correlated with each
other and the resulting log-returns are bivariately normally distributed
subject to 
\[
\left(\begin{array}{c}
\ln\left(\frac{S_{1}(t)}{S_{1}(0)}\right)\\
\ln\left(\frac{S_{2}(t)}{S_{2}(0)}\right)
\end{array}\right)\sim N\left(\left(\begin{array}{c}
\mu_{1}-\frac{1}{2}\sigma_{1}^{2}\\
\mu_{2}-\frac{1}{2}\sigma_{2}^{2}
\end{array}\right)\cdot t,\bm{C}\cdot t\right),
\]
and variance-covariance matrix 
\[
\bm{C}=\left(\begin{array}{cc}
\sigma_{1}^{2} & \rho_{12}\sigma_{1}\sigma_{2}\\
\rho_{12}\sigma_{1}\sigma_{2} & \sigma_{2}^{2}
\end{array}\right).
\]

Within the scope of this paper we restrict ourselves to self-financing
strategies, that is strategies where money is neither injected nor
withdrawn from the portfolio during the trading period $(0,T)$. Moreover,
we are focussing on performance participation strategies that are
built on the two risky assets $S_{1}$, $S_{2}$ only. Following \citet{BlackScholes1973}
the underlying market is assumed to provide the usual perfect market
conditions including no arbitrage and completeness.%
\footnote{See, e.g., \citet{BlackScholes1973} or \citet{Shreve2008}.%
}

As introduced in Section \ref{sec:Introduction}, performance participation
strategies are investment strategies built on the two risky assets
$S_{1}$, $S_{2}$ that provide a minimum performance in terms of
a fraction of the outcome of the reserve asset $S_{1}$ while keeping
the potential for profits resulting from the outperformance of the
active asset $S_{2}$. To facilitate a return perspective now and
in the following we assume w.l.o.g. that the initial values of both
risky underlyings equal the initial portfolio value $V_{0}$, i.e.
$S_{1}(0)=S_{2}(0)=V_{0}.$

The next sections provide a formalized and unified definition of the
two performance participation strategies. We start with the definition
of the OBPP strategy as a static example of a performance participation
trading rule.

\subsection{The Option-Based Performance Participation (OBPP) strategy\label{sec:OBPP}}

The Option-Based Performance Participation (OBPP) strategy generates
the desired participation with the aid of exchange options. An exchange
option gives the option holder the right to exchange at its expiry
one risky asset for another. \citet{Margrabe1978} was the first to
introduce and develop an equation for the value of an exchange option.
Let $T$ denote the terminal trading date. The minimum terminal wealth
which must be achieved is given by the fraction $\alpha<1$ of the
performance of the reserve asset $S_{1}$ at maturity $T$, i.e. 
\begin{equation}
F(T)=\alpha\cdot S_{1}(T).\label{eq:floor}
\end{equation}
In analogy to standard portfolio insurance strategies we denote the
current value of the (stochastic) performance participation $F(t)=\alpha\cdot S_{1}(t)$
the floor.

Thus, purchasing at inception $t=0$ an adequate number of shares
$p$ of the active asset $S_{2}$ and one exchange option written
on $\alpha$ shares of the reserve asset $S_{1}$ and $p$ shares
of $S_{2}$, respectively, enables the desired performance participation.
Note that the dampening factor $p<1$ is related to the value of the
exchange option and thus reflects the costs of the desired performance
guarantee. It will be analyzed in more detail later on.

More precisely, given the payoff of the exchange option at maturity
$T$ 
\begin{equation}
V^{ex}\left(T;T,\alpha\cdot S_{1},p\cdot S_{2}\right)=\left(\alpha\cdot S_{1}(T)-p\cdot S_{2}(T)\right)^{+},\label{eq:payoff_Ex}
\end{equation}
the obtained terminal portfolio value of the OBPP strategy then yields
\begin{align}
V^{OBPP}\left(T;T,V_{0},\alpha,S_{1},S_{2}\right) & =p\cdot S_{2}\left(T\right)+\left(\alpha\cdot S_{1}\left(T\right)-p\cdot S_{2}\left(T\right)\right)^{+}\label{eq:payoff_OBPP}\\
 & =\max\left\{ \alpha\cdot S_{1}\left(T\right),p\cdot S_{2}\left(T\right)\right\} \geq\alpha\cdot S_{1}\left(T\right).
\end{align}
 Hence, additionally to the guaranteed wealth $\alpha\cdot S_{1}(T)$
a participation - at a percentage $p$ - in the outperformance of
the active asset $S_{2}$ is possible. The obtained payoff is the
maximum of the stochastic floor $\alpha\cdot S_{1}$ and the down-scaled
performance of the active asset $S_{2}$. Thus, within the context
of the OBPP strategy the purchased exchange option can be interpreted
as a protecting put option with stochastic strike $\alpha\cdot S_{1}$.%
\footnote{Note that \citet{Margrabe1978} was the first to use this interpretation.%
} Following from put-call-parity for exchange options%
\footnote{See \citet{Margrabe1978}.%
} the portfolio setup (\ref{eq:payoff_OBPP}) is furthermore equivalent
to holding the stochastic floor $F(t)=\alpha\cdot S_{1}(t)$ plus
the exchange option $V^{ex}\left(t;T,p\cdot S_{2},\alpha\cdot S_{1}\right)$
that gives the option holder at its maturity $T$ the right to exchange
$\alpha$ shares of the reserve asset $S_{1}$ against $p$ shares
of the active asset $S_{2}$. With this respect, the exchange option
plays the role of a call option written on the scaled underlying $p\cdot S_{2}$
with stochastic strike $\alpha\cdot S_{1}(T)$.

The percentage $p$ of the active asset is derived in such a way to
match the investor's initial endowment $V_{0}$ and insurance needs
$\alpha$. More precisely, at inception $t=0$ the initial capital
is adequately split to purchase both $\alpha$ shares of $S_{1}$
representing the stochastic floor $F$ and the protecting exchange
option $V^{ex}\left(t;T,p\cdot S_{2},\alpha\cdot S_{1}\right)$. This
implies the condition%
\footnote{Note that Equation (\ref{eq:p_cond}) can be solved for the adequate
percentage $p$ using standard zero search methods like, e.g., the
Newton gradient method. Furthermore, it only possesses a solution
in $p$ if we assume $\alpha<1$. This solution will be unique as
the value of the exchange option and thus the initial OBPP portfolio
value are strictly monotone in $p$. In case that $\alpha\geq1$ and
substituting $S_{1}(0)=S_{2}(0)=V_{0}$ Equation (\ref{eq:p_cond})
yields 
\[
\alpha\cdot S_{1}(0)+V^{ex}\left(0;T,p\cdot S_{2},\alpha\cdot S_{1}\right)-V_{0}\geq V^{ex}\left(0;T,p\cdot S_{2},\alpha\cdot S_{1}\right)>0,
\]
and there will be no solution.%
} 
\begin{equation}
V_{0}=V^{OBPP}\left(0;T,V_{0},\alpha,S_{1},S_{2}\right)=\alpha\cdot S_{1}(0)+V^{ex}\left(0;T,p\cdot S_{2},\alpha\cdot S_{1}\right).\label{eq:p_cond}
\end{equation}
Note that since the value of the exchange option is always positive,
the put-call-parity for exchange options directly induces the upper
bound $p<1$.

The OBPP is designed as a static investment strategy.%
\footnote{Note that in practice the underlying exchange option will usually
be dynamically replicated. This synthesized OBPP represents a dynamic
strategy as well. For further details we refer the interested reader
to Section \ref{sec:DynamicBehavior}.%
} Hence, once allocated the portfolio constitution remains unchanged
during the investment period $(0,T)$. By applying \citet{Margrabe1978}'s
formula for the price of the exchange option the current value of
the OBPP portfolio at any time $t\in\left[\left.0,T\right)\right.$
is given by 
\begin{equation}
V^{OBPP}\left(t;T,V_{0},\alpha,S_{1},S_{2}\right)=\alpha\cdot S_{1}(t)+V^{ex}\left(t;T,p\cdot S_{2},\alpha\cdot S_{1}\right),\label{eq:value_OBPP}
\end{equation}
where 
\begin{equation}
V^{ex}\left(t;T,p\cdot S_{2},\alpha\cdot S_{1}\right)=p\cdot S_{2}(t)\cdot\Phi\left(d_{1}\right)-\alpha\cdot S_{1}(t)\cdot\Phi\left(d_{2}\right),\label{eq:value_Ex}
\end{equation}
and 
\begin{align}
d_{1} & =\frac{\ln\left(\frac{p\cdot S_{2}(t)}{\alpha\cdot S_{1}(t)}\right)+\frac{1}{2}\cdot\hat{\sigma}_{2}^{2}\cdot(T-t)}{\hat{\sigma}_{2}\cdot\sqrt{T-t}},\label{eq:d_1_ex}\\
d_{2} & =d_{1}-\hat{\sigma}_{2}\cdot\sqrt{T-t}.\label{eq:d_2_ex}
\end{align}
Here, $\Phi$ denotes the cumulative distribution function of the
standard normal distribution. The constant $\hat{\sigma}_{2}$ given
by 
\begin{equation}
\hat{\sigma}_{2}^{2}=\sigma_{1}^{2}-2\cdot\rho_{12}\cdot\sigma_{1}\cdot\sigma_{2}+\sigma_{2}^{2},\label{eq:sigma_2_hat}
\end{equation}
is the volatility of the ratio process%
\footnote{See \citep{Margrabe1978} or later on Remark \ref{rmk:SDE_S_2_hat}.%
} $S_{2}/S_{1}$. Since it is a decreasing function in the correlation
$\rho_{12}$, the protecting exchange option is the cheaper the higher
the correlation between the two underlyings. A high correlation signifies
a likewise simultaneous evolution of the risky assets. Thus, the probability
that the option will be executed at maturity is reduced.

Note that since the value of the exchange option is always positive,
for all dates $t\in[0,T]$ the portfolio value is actually always
above the floor $F(t)$. The desired minimum performance participation
is thus not only provided at the terminal date $T$ but also on an
intertemporal basis. Furthermore, the portfolio weights of the corresponding
replicating strategy are always smaller or equal to one. The OBPP
strategy is thus leveraging neither of the two underlyings. As we
will see in the sequel, this is one of the main differences between
the two performance participation strategies under consideration.

As we have mentioned above, the constant $\hat{\sigma}_{2}$ represents
the diffusion of the process $S_{2}/S_{1}$, which we denote now and
in the following by $\hat{S}_{2}=S_{2}/S_{1}$. With respect to that
asset ratio, also called the index ratio%
\footnote{The notation goes back to \citep{BlackPerold1992}.%
}, we can establish a very important relationship between the newly
introduced OBPP and the standard OBPI strategy. 
\begin{lem}[OBPP and OBPI value]
\label{lem:OBPP_OBPI} Given the financial market $\left(\Omega,\mathcal{F},\mathbb{P},\mathbb{F}\right)$
defined in (\ref{eq:SDE_risky}) and the risky asset $S\in\left\{ S_{1},S_{2}\right\} $.
Furthermore, let $T$ denote the horizon of the desired insurance
level $\alpha^{PI}\leq e^{r\cdot T}$ in terms of the initial endowment
$V_{0}$ of a standard OBPI strategy, whose current portfolio value
at time $t\in[0,T]$ is given by%
\footnote{See, e.g., \citet{BertrandPrigent2005}.%
} 
\begin{equation}
V^{OBPI}\left(t;T,V_{0},\alpha^{PI},r,S\right)=\alpha^{PI}\cdot V_{0}\cdot e^{-r\cdot(T-t)}+Call\left(t;T,\alpha^{PI}\cdot V_{0},r,\sigma_{S},p^{PI}\cdot S\right).\label{eq:value_OBPI}
\end{equation}
Here, $Call\left(t;T,\alpha^{PI}\cdot V_{0},r,\sigma_{S},p^{PI}\cdot S\right)$
denotes the Black-Scholes value of a vanilla call option at time $t$,
with maturity $T$, written on $p^{PI}$ shares of the risky asset
$S$ with strike $\alpha^{PI}\cdot V_{0}$, risk-free interest rate
$r$ and volatility $\sigma_{S}$. The number of shares $p^{PI}<1$
of the risky underlying $S$ is adapted to the desired terminal guarantee
$\alpha^{PI}\cdot V_{0}$ and the initial endowment $V_{0}$ via the
condition 
\begin{equation}
V_{0}=\alpha^{PI}\cdot V_{0}\cdot e^{-r\cdot T}+Call\left(0;T,\alpha^{PI}\cdot V_{0},r,\sigma_{S},p^{PI}\cdot S\right).\label{eq:p_cond_OBPI}
\end{equation}
Then, at any time $t\in[0,T]$ the OBPP strategy can be represented
as a portfolio consisting of $S_{1}(t)$ shares of a standard OBPI
strategy in the discounted market with $S_{1}$ as numéraire 
\begin{equation}
V^{OBPP}\left(t;T,V_{0},\alpha,S_{1},S_{2}\right)=S_{1}(t)\cdot\hat{V}^{OBPI}\left(t;T,\hat{V}_{0},\alpha,\hat{r},\hat{S}_{2}\right).\label{eq:OBPP_OBPI}
\end{equation}
Note that in the discounted market with $S_{1}$ as numéraire we thus
consider the discounted assets $S_{1}/S_{1}$ and $S_{2}/S_{1}$.
Whereas the former is constant, yielding the risk-free interest rate
$\hat{r}=0$, the later represents the index ratio $\hat{S}_{2}$.
The same applies to the initial portfolio value that reduces to $\hat{V}_{0}=1$.
All other parameters remain the same. \end{lem}
\begin{proof}
The relationship can be easily derived by observing that the discounted
exchange option $V^{ex}\left(t;T,p\cdot S_{2},\alpha\cdot S_{1}\right)$
with respect to $S_{1}$ as numéraire is equivalent to a standard
vanilla call option written on $p$ shares of the risky underlying
$p\cdot\hat{S}_{2}$, with strike $\alpha$ and risk-free interest
rate $\hat{r}=0$. More precisely, let $\hat{\mathbb{P}}_{i}$, $i=0,1$
denote the equivalent martingale measure corresponding to the numéraire
$S_{i}$. Then, following from the risk-neutral pricing formula and
the change of numéraire theorem%
\footnote{See, e.g., \citet{Shreve2008}.%
} we obtain for the value of the exchange option $V^{ex}\left(t;T,p\cdot S_{2},\alpha\cdot S_{1}\right)$
at time $t\in[0,T]$ 
\begin{align*}
V^{ex}\left(t;T,p\cdot S_{2},\alpha\cdot S_{1}\right) & =S_{0}(t)\cdot\mathbb{E}_{\hat{\mathbb{P}}_{0}}\left[\left.\left(p\cdot\frac{S_{2}(T)}{S_{0}(T)}-\alpha\cdot\frac{S_{1}(T)}{S_{0}(T)}\right)^{+}\right|\mathcal{F}_{t}\right]\\
 & =S_{1}(t)\cdot\mathbb{E}_{\hat{\mathbb{P}}_{1}}\left[\left.\left(p\cdot\hat{S}_{2}(T)-\alpha\right)^{+}\right|\mathcal{F}_{t}\right].
\end{align*}
Since the value of a call option only depends on the risk-free interest
rate as well as the volatility of the risky asset, yielding $\hat{r}=0$
and $\hat{\sigma}_{2}$ in the discounted market, we conclude that
\[
\mathbb{E}_{\hat{\mathbb{P}}_{1}}\left[\left.\left(p\cdot\hat{S}_{2}(T)-\alpha\right)^{+}\right|\mathcal{F}_{t}\right]=Call\left(t;T,\alpha,\hat{r},\hat{\sigma}_{2},p\cdot\hat{S}_{2}\right).
\]
Overall this leads to%
\footnote{Note that this equality was already shown in \citet{Margrabe1978}
using a different motivation.%
} 
\begin{equation}
\frac{V^{ex}\left(t;T,p\cdot S_{2},\alpha\cdot S_{1}\right)}{S_{1}(t)}=Call\left(t;T,\alpha,\hat{r},\hat{\sigma}_{2},p\cdot\hat{S}_{2}\right).\label{eq:Ex_discounted_Call}
\end{equation}
Note that due to this discounting property the adequate number of
shares $p$ is the same for the OBPP and the OBPI.

Hence, the additionally introduced source of risk in terms of a risky
reserve asset manifests itself as stochastic numéraire that allows
to reduce the newly introduced performance participation strategy
to its portfolio insurance equivalent in the discounted asset universe.
The stochastic dynamics of the index ratio $\hat{S}_{2}$ are provided
in the following remark. \end{proof}
\begin{rem}
\label{rmk:SDE_S_2_hat} Define the value process of the ratio of
the two risky assets $\hat{S}_{2}=S_{2}/S_{1}$. The process $\hat{S}_{2}$
is lognormal and given by the geometric Brownian motion 
\begin{equation}
d\hat{S}_{2}(t)=\hat{S}_{2}(t)\hat{\mu}_{2}dt+\hat{S}_{2}(t)\hat{\sigma}_{2}dW_{\hat{S}_{2}}(t),\ \ \hat{S}_{2}(0)=1,\label{eq:SDE_S_2_hat}
\end{equation}
with drift 
\begin{equation}
\hat{\mu}_{2}=\left(\mu_{2}-\mu_{1}\right)+\left(\sigma_{1}^{2}-\rho_{12}\sigma_{1}\sigma_{2}\right),\label{eq:mu_2_hat}
\end{equation}
volatility $\hat{\sigma}_{2}$ as defined in Equation (\ref{eq:sigma_2_hat})
and Wiener process 
\begin{equation}
W_{\hat{S}_{2}}=\frac{\rho_{12}\cdot\sigma_{2}-\sigma_{1}}{\hat{\sigma}_{2}}\cdot W_{1}+\frac{\sqrt{1-\rho_{12}^{2}}\cdot\sigma_{2}}{\hat{\sigma}_{2}}\cdot W_{2}.\label{eq:W_S_2_hat}
\end{equation}
\end{rem}
\begin{proof}
The stochastic dynamics follow directly from Itô's lemma and the one-dimensional
Lévy theorem.%
\footnote{See, e.g., \citet{Shreve2008}.%
}
\end{proof}
To conclude the section we analyze the additional scaling factor $p$
in more detail. As mentioned earlier, it is necessary to provide arbitrary
investor-specific levels of performance participation $\alpha<1$.
The OBPP is thus a generalization of the earlier mentioned Best of
Two strategy 
\[
V^{Bo2}\left(T;T,V_{0},S_{1},S_{2}\right)=p^{Bo2}\cdot\max\left\{ S_{1}(T),S_{2}(T)\right\} ,
\]
that (except for the factor $p^{Bo2}$) returns the better performing
underlying at the end of the investment horizon $T$. It corresponds
to the special case of the OBPP where $\alpha=p^{Bo2}$. Note that
the factor $p^{Bo2}<1$ cannot be omitted and is necessary to adjust
the portfolio allocation to the prespecified initial endowment $V_{0}$.

With respect to arbitrary participation levels the percentage $p$
is a decreasing function of $\alpha<1$. For this purpose we recall
the initial endowment Condition (\ref{eq:p_cond}) 
\[
V_{0}-\alpha\cdot S_{1}(0)=V^{ex}\left(0;T,p\cdot S_{2},\alpha\cdot S_{1}\right),
\]
or following from put-call-parity for exchange options equivalently
\[
V_{0}-p\cdot S_{2}(0)=V^{ex}\left(0;T,\alpha\cdot S_{1},p\cdot S_{2}\right),
\]
where the left-hand side is decreasing in $p$ whereas the value of
the exchange option is increasing in $\alpha$ and decreasing in $p$,
respectively.

Note that in the special case where the reserve asset is given by
a zero-coupon bond with face value $V_{0}$, i.e. $S_{1}(t)=V_{0}\cdot e^{-r\cdot(T-t)}$,
the exchange option $V^{ex}\left(t;T,p\cdot S_{2},\alpha\cdot S_{1}\right)$
with risk-free asset $S_{1}$ reduces to a standard vanilla call option
written on $p\cdot S_{2}$ with strike $\alpha\cdot S_{1}(T)=\alpha\cdot V_{0}$,
i.e.%
\footnote{See \citet{Margrabe1978}.%
} 
\[
V^{ex}\left(t;T,p\cdot S_{2},\alpha\cdot S_{1}\right)=Call\left(t;T,\alpha\cdot V_{0},r,\sigma_{2},p\cdot S_{2}\right).
\]
The OBPP strategy with level of performance participation $\alpha$
then represents a standard OBPI strategy with respect to the deterministic
insurance level $\alpha^{PI}=\alpha$.

In the next section we will elaborate on \citet{BlackPerold1992}'s
idea of a CPPI strategy defined on a portfolio consisting of two risky
assets. Since their original approach to portfolio insurance with
a risky reserve asset is not widely spread in the literature we will
redefine it as a dynamic approach to the more general class of performance
participation strategies. Furthermore, the name will be adapted to
cope with the more general performance participation feature.

\subsection{The Constant Proportion Performance Participation (CPPP) strategy\label{sec:CPPP}}

Similar to the OBPP strategy the Constant Proportion Performance Participation
(CPPP) strategy aims at providing a minimum return participation in
the reserve asset $S_{1}$ while benefiting from an outperformance
of the active asset $S_{2}$. This is achieved by applying the CPPI
investment rules to a portfolio consisting of two risky assets. In
contrast to the OBPP strategy the CPPP represents a dynamic strategy
since the portfolio is continuously reallocated over time. Furthermore,
the applied allocation rules even allow for a leveraged participation
in $S_{2}$.

Let again $\alpha<1$ denote the minimum investor-defined level of
performance participation in the risky reserve asset $S_{1}$ that
defines the portfolio floor $\left(F(t)\right)_{0\leq t\leq T}$,
i.e. 
\[
F(t)=\alpha\cdot S_{1}(t).
\]
This minimum portfolio value has to be achieved not only at the end
of the investment horizon $T$ but at any time $t\in[0,T].$ Furthermore,
we define at time $t\in[0,T]$ the cushion $C$ as the excess portfolio
value with respect to the current floor 
\[
C(t)=\max\left\{ V^{CPPP}(t)-F(t),0\right\} .
\]
Note that the requirement of a positive initial cushion $C_{0}=V_{0}-F_{0}$,
where $V_{0}=S_{1}(0)$, establishes the natural bound $\alpha<1$
on the level of performance participation. In order to ensure the
required floor $F(t)$ at any time $t\in[0,T]$ the basic idea of
the CPPP method now consists in analogy to the standard CPPI strategy
in investing a constant proportion $m>0$ of the cushion $C$ in the
active asset $S_{2}$. This is the reason why we call the strategy
constant proportion performance participation. The remaining part
of the portfolio is invested in the reserve asset $S_{1}$. More precisely,
the exposures $E_{2}$ and $E_{1}$ to the active and the reserve
asset $S_{2}$, $S_{1}$, respectively, at time $t\in[0,T]$ are determined
by 
\begin{align*}
E_{2}(t) & =m\cdot C(t)=m\cdot\max\left\{ V^{CPPP}(t)-F(t),0\right\} ,\\
E_{1}(t) & =V^{CPPP}(t)-E_{2}(t).
\end{align*}
The constant multiplier $m$ affects the participation in the (out)performance
of asset $S_{2}$ and the potential leverage effect with respect to
$S_{1}$. In general, the strategy is well-defined for any $m>0$.
However, we will restrict to the more interesting case $m\geq1$ when
the payoff function is convex in the value of the active asset $S_{2}$.

In their seminal paper \citet{BlackPerold1992} already derive the
value of the CPPP portfolio by establishing a similar relationship
with the standard CPPI strategy as it is the case for OBPP and OBPI
according to Equation (\ref{eq:OBPP_OBPI}). 
\begin{lem}[CPPP and CPPI value]
\label{lem:CPPP_CPPI} Given the financial market $\left(\Omega,\mathcal{F},\mathbb{P},\mathbb{F}\right)$
defined in (\ref{eq:SDE_risky}) and the risky asset $S\in\left\{ S_{1},S_{2}\right\} $.
Furthermore, let $T$ denote the horizon of the desired insurance
level $\alpha^{PI}\leq e^{r\cdot T}$ in terms of the initial endowment
$V_{0}$ of a standard CPPI strategy with multiplier $m$, whose current
portfolio value at time $t\in[0,T]$ is given by%
\footnote{See, e.g., \citet{PeroldSharpe1988}.%
} 
\begin{align}
V^{CPPI}\left(t;T,V_{0},\alpha^{PI},m,r,S\right) & =\alpha^{PI}\cdot V_{0}\cdot e^{-r\cdot(T-t)}\label{eq:value_CPPI}\\
 & +\beta_{CPPI}\left(t;\alpha^{PI},m,r,\sigma_{S}\right)\cdot V_{0}\cdot e^{r\cdot t}\cdot\left(\frac{S(t)}{V_{0}\cdot e^{r\cdot t}}\right)^{m},\nonumber 
\end{align}
with the non-negative function $\beta_{CPPI}\left(t;\alpha^{PI},m,r,\sigma_{S}\right)$
defined as 
\[
\beta_{CPPI}\left(t;\alpha^{PI},m,r,\sigma_{S}\right)=\left(1-\alpha^{PI}\cdot e^{-r\cdot T}\right)\cdot e^{\frac{1}{2}\cdot m\cdot(1-m)\cdot\sigma_{S}^{2}\cdot t}.
\]
 Then, at any time $t\in[0,T]$ the CPPP strategy can be represented
as a portfolio consisting of $S_{1}(t)$ shares of a standard CPPI
strategy in the discounted market with $S_{1}$ as numéraire 
\begin{equation}
V^{CPPP}\left(t;T,V_{0},\alpha,m,S_{1},S_{2}\right)=S_{1}(t)\cdot\hat{V}^{CPPI}\left(t;T,\hat{V}_{0},\alpha,m,\hat{r},\hat{S}_{2}\right),\label{eq:CPPP_CPPI}
\end{equation}
where $\hat{r}=0$ and $\hat{S}_{2}=S_{2}/S_{1}$. All other parameters
remain the same. More precisely, the current CPPP portfolio value
is calculated as 
\begin{align}
V^{CPPP}\left(t;T,V_{0},\alpha,m,S_{1},S_{2}\right) & =F(t)+C(t)\label{eq:value_CPPP}\\
 & =\alpha\cdot S_{1}(t)+\beta_{CPPP}\left(t;\alpha,m,\hat{\sigma}_{2}\right)\cdot S_{1}(t)\cdot\left(\frac{S_{2}(t)}{S_{1}(t)}\right)^{m},\nonumber 
\end{align}
where $\beta_{CPPP}\left(t;\alpha,m,\hat{\sigma}_{2}\right)=\beta_{CPPI}\left(t;\alpha,m,\hat{r},\hat{\sigma}_{2}\right)=\left(1-\alpha\right)\cdot e^{\frac{1}{2}\cdot m\cdot(1-m)\cdot\hat{\sigma}_{2}^{2}\cdot t}$. \end{lem}
\begin{proof}
Recall that a change of numéraire does not affect the underlying self-financing
CPPP investment rule.%
\footnote{See, e.g., \citet{Shreve2008}.%
} Thus, the number of shares $\varphi_{i}(t)$ allocated of asset $S_{i}$,
$i=1,2$ at time $t\in[0,T]$ in the original (denoted by $V^{CPPP}(t)$)
and the discounted CPPP portfolio%
\footnote{Note that for clearness we sometimes omit the detailed declaration
of all parameters of the performance participation strategy PP and
simply denote the current portfolio value by $V^{PP}$.%
} (denoted by $\hat{V}^{CPPP}(t)$) are the same and yield 
\[
\varphi_{1}(t)=\frac{V^{CPPP}(t)-m\cdot\left(V^{CPPP}(t)-F(t)\right)}{S_{1}(t)}\ \ \text{and}\ \ \varphi_{2}(t)=\frac{m\cdot\left(V^{CPPP}(t)-F(t)\right)}{S_{2}(t)}.
\]
This can be further transformed to 
\[
\varphi_{1}(t)=\frac{\hat{V}^{CPPP}(t)-m\cdot\left(\hat{V}^{CPPP}(t)-\alpha\cdot1\right)}{1}\ \ \text{and}\ \ \varphi_{2}(t)=\frac{m\cdot\left(\hat{V}^{CPPP}(t)-\alpha\cdot1\right)}{\hat{S}_{2}(t)}.
\]
which actually represents a standard CPPI strategy with respect to
the risk-free interest rate $\hat{r}$ and the index ratio $\hat{S}_{2}$.%
\footnote{Note that this relationship was already stated in \citet{BlackPerold1992}.%
} Equation (\ref{eq:value_CPPP}) then follows directly from (\ref{eq:CPPP_CPPI})
by substituting $\alpha^{PI}=\alpha$, $\hat{V}_{0}=1$, $\hat{r}=0$,
$\hat{\sigma}_{2}$ and $\hat{S}_{2}$ in (\ref{eq:value_CPPI}). \end{proof}
\begin{rem}[Cushion dynamics]
The cushion process $C$ of the CPPP is lognormal and given by 
\begin{equation}
dC(t)=C(t)\mu_{C}dt+C(t)\sigma_{C}dW_{C}(t),\label{eq:SDE_C}
\end{equation}
with mean rate of return and volatility 
\begin{align}
\mu_{C} & =\mu_{1}+m\cdot\left(\mu_{2}-\mu_{1}\right),\label{eq:mu_C}\\
\sigma_{C}^{2} & =(1-m)^{2}\cdot\sigma_{1}^{2}+2\cdot(1-m)\cdot m\cdot\rho_{12}\cdot\sigma_{1}\cdot\sigma_{2}+m^{2}\cdot\sigma_{2}^{2}.\label{eq:sigma_C}
\end{align}
\end{rem}
\begin{proof}
The stochastic dynamics of $C$ follow by application of Itô's lemma
and the one-dimensio\-nal Lévy theorem.
\end{proof}
Hence, similar to the OBPP, the additional source of risk in terms
of a risky reserve asset manifests itself as stochastic numéraire
that allows to reduce the newly introduced performance participation
strategy to its portfolio insurance equivalent in the discounted asset
universe. In the sequel the derived relationships (\ref{eq:OBPP_OBPI})
and (\ref{eq:CPPP_CPPI}) will be very useful for the analysis of
the characteristics of the two performance participation strategies.
Especially with respect to the moments of the resulting payoff distributions
as well as the dynamic behavior it allows to perform most of the examinations
in terms of the standard portfolio insurance strategies in the reduced
discounted market framework. The main benefit being that the latter
strategies have already been extensively studied from an analytical
point of view.%
\footnote{See, e.g., \citet{BlackRouhani1989}, \citet{BlackPerold1992}, \citet{BertrandPrigent2005}
or \citet{ZagstKraus2009}.%
}

Equation (\ref{eq:value_CPPP}) represents the basic properties of
the CPPP. At any time $t$ the value of the strategy consists of the
current value of the guarantee $F(t)$ and the strictly positive cushion
$C(t)$ which is proportional to $S_{1}$ and $(S_{2}/S_{1})^{m}$.
Thus, the CPPP value always lies strictly above the dynamically insured
floor $F(t)$. Furthermore, the CPPP value process is path independent.

In contrast to the OBPP approach the CPPP includes an additional degree
of freedom which is introduced by the multiplier $m$. The payoff
above the stochastic guarantee, i.e. the cushion, is linear in $S_{2}$
for $m=1$ and it is convex in $S_{2}$ (and $S_{2}/S_{1}$) for $m>1$.
In the latter case the resulting exposure to the active asset $S_{2}$
is likely to exceed the actual portfolio value. This is due to the
leveraging effect associated with $m$. The exposure to asset $S_{2}$
is then financed by short-selling the reserve asset $S_{1}$.

Note that in the special case when the reserve asset is given by a
zero-coupon bond with face value $V_{0}$, i.e. $S_{1}(t)=V_{0}\cdot e^{-r\cdot(T-t)}$,
the CPPP strategy with level of performance participation $\alpha$
reduces to a standard CPPI strategy with respect to the deterministic
insurance level $\alpha^{PI}=\alpha$.

In what follows we compare the two performance participation strategies
with respect to various criteria including moments as well as the
dynamic behavior.

\section{Comparison of the Payoff Distributions}

In order to compare the two methods we retain the assumption that
the initial portfolio values are the same and equal the initial asset
prices, i.e. 
\[
V_{0}=V^{CPPP}(0)=V^{OBPP}(0)=S_{1}(0)=S_{2}(0).
\]
Furthermore, the two strategies are supposed to provide the same participation
$\alpha<1$ in the performance of the reserve asset $S_{1}$. Hence,
\[
F(t)=\alpha\cdot S_{1}(t),
\]
in the case of the CPPP strategy and the adequate number of shares
$p$ of the OBPP strategy is derived from Condition (\ref{eq:p_cond})
\[
V_{0}=\alpha\cdot S_{1}(0)+V^{ex}\left(0;T,p\cdot S_{2},\alpha\cdot S_{1}\right).
\]
Note that these two conditions do not impose any constraint on the
multiplier $m$ as the second parameter of the CPPP strategy. In what
follows, this leads us to consider CPPP strategies for various values
of the multiplier; Among them the unique value $m^{*}$ which complies
with equality of payoff expectations as an additional condition (see
Section \ref{sec:equal_payoffExpectations} for details). We start
with the analysis of the payoff functions of both strategies.

\subsection{Comparison of the payoff functions}

In the simplest case one of the payoff functions of the two methods
would statewisely dominate the other one. More precisely, this implies
that one of the portfolio values always lies above the other one for
all terminal values $S_{1}(T)$, $S_{2}(T)$. However, since $V_{0}=V^{CPPP}(0)=V^{OBPP}(0)$
and due to the absence of arbitrage this is not possible.%
\footnote{See \citet{BlackRouhani1989}.%
} 
\begin{lem}
Neither of the two payoffs is greater than the other one for all terminal
values $S_{1}(T)$, $S_{2}(T)$ of the underlying risky assets. The
two payoff functions thus intersect one another.\end{lem}
\begin{proof}
The proposition follows together with Equation (\ref{eq:OBPP_OBPI})
and (\ref{eq:CPPP_CPPI}) from the analog relationship with respect
to the standard OBPI and CPPI strategy which was shown in \citet{ZagstKraus2009}.
\end{proof}
Figure \ref{fig:statewiseDominance} illustrates this finding using
a simple numerical example with typical values for the financial market
presented in Table \ref{tab:StandardParameterSet}.%
\footnote{The asset characteristics were obtained from monthly return data of
the JP Morgan EMU Government Bond Index and the Dow Jones Eurostoxx
50 Index over the time period 01/1995-10/2009.%
} Furthermore, the adequate number of shares and exchange options $p$
corresponding to the initial endowment $V_{0}=100$ and an investor-defined
level of performance participation $\alpha=0.95$ are provided. If
not mentioned otherwise, now and in the following we will consider
this setting as our reference model scenario for numerical calculations.
\begin{table}[htbp]
\centering %
\begin{tabular}{l|cc||l|c}
Market parameters  & Reserve asset $S_{1}$  & Active asset $S_{2}$  & Strategy parameters  & \tabularnewline
\hline 
\hline 
$\mu_{i}$  & 6.6\%  & 9.7\%  & $V_{0}$  & 100 \tabularnewline
$\sigma_{i}$  & 3.7\%  & 21.4\%  & $T$  & 1 (year) \tabularnewline
$\rho_{12}$  & \multicolumn{2}{c||}{-0.15} & $\alpha$  & 0.95 \tabularnewline
$\hat{\sigma}_{2}$  & \multicolumn{2}{c||}{22.3\%} & $p$  & 0.8780 \tabularnewline
\end{tabular}\caption{Standard parameter set for the financial market as well as the two
performance participation strategies under consideration.}

\label{tab:StandardParameterSet} 
\end{table}

\begin{figure}[htbp]
\centering \includegraphics[width=0.7\textwidth]{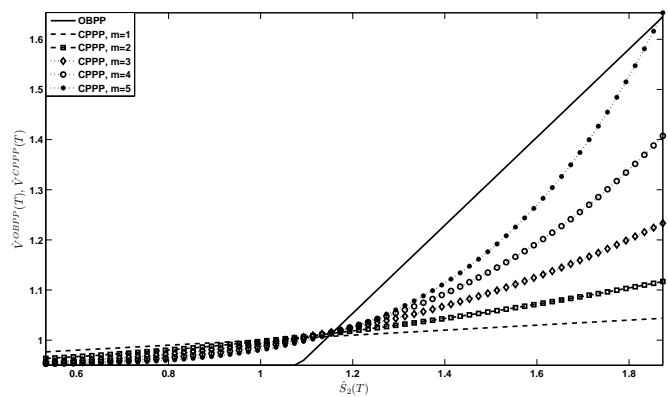}
\caption{CPPP and OBPP payoffs as functions of $\hat{S}_{2}(T)$ according
to (\ref{eq:CPPP_CPPI}) and (\ref{eq:OBPP_OBPI}) for the standard
parameter set provided in Table \ref{tab:StandardParameterSet} and
$m=1,2,...,5$.}

\label{fig:statewiseDominance} 
\end{figure}

The graph visualizes the strategy payoffs according to (\ref{eq:OBPP_OBPI})
and (\ref{eq:CPPP_CPPI}) as functions of the terminal index ratio
$\hat{S}_{2}(T)$, i.e. the standard OBPI and CPPI in the discounted
market. With respect to the CPPP strategy different values of the
multiplier $m=1,...,5$ are analyzed. For each value of the multiplier
the payoffs intersect at least once. The CPPP payoff exceeds the OBPP
one not only for very large values of the index ratio but also for
the more important range of moderate outperformance and even underperformance
of the active asset with respect to the reserve asset. For $m=1$
it is a linear and for $m>1$ an exponential function of $\hat{S}_{2}$.
As the value of the multiplier increases, the CPPP portfolio value
becomes more convex in $\hat{S}_{2}$. In contrast, the OBPP payoff
is a (piecewise) linear function of the terminal index ratio.

The examination of the terminal performances is only a first step
within the scope of a comparison of the two strategies. However, a
detailed analysis must also take into account the entire payoff distributions
- including the probabilities of bullish and bearish markets. In the
sequel we will thus derive explicit formulas for the moments of the
resulting distributions. This enables us to extend the analysis especially
to the first four moments.

\subsection{Comparison of the moments of the payoff distributions}

\subsubsection{Moments of the CPPP and the OBPP strategy}

\label{sec:Moments}

To derive explicit formulas for the moments of the payoff distributions
of the OBPP and the CPPP we will make use of the similarity of performance
participation and portfolio insurance strategies according to (\ref{eq:OBPP_OBPI})
and (\ref{eq:CPPP_CPPI}). As a byproduct we obtain general formulas
for the moments of the standard OBPI and CPPI, too.
\begin{lem}
\label{lem:kth_Moment_PP}Let $V^{PP}(t)$ denote the portfolio value
of the OBPP or the CPPP strategy at time $t\in[0,T]$ and $\hat{V}^{PI}(t)$
the respective value of the corresponding portfolio insurance strategy
in the discounted market according to (\ref{eq:OBPP_OBPI}) and (\ref{eq:CPPP_CPPI}).
Then, the $k$th moment $m_{k}\left(V^{PP}(t)\right)$, $k\in\mathbb{N}$
of the performance participation strategy PP with respect to the real-world
measure $\mathbb{P}$ can be calculated as 
\begin{equation}
m_{k}\left(V^{PP}(t)\right)=\mathbb{E}_{\mathbb{P}}\left[S_{1}(t)^{k}\right]\cdot\tilde{m}_{k}\left(\hat{V}^{PI}(t)\right),\label{eq:kth_Moment_PP_PI}
\end{equation}
where $\tilde{m}_{k}\left(\hat{V}^{PI}(t)\right)$ denotes the $k$th
moment of the associated portfolio insurance strategy with respect
to the equivalent probability measure $\tilde{\mathbb{P}}_{k}$ defined
by the Radon-Nikodym derivative%
\footnote{See, e.g., \citet{Shreve2008}.%
} 
\begin{equation}
\begin{array}{cc}
\left.\frac{d\tilde{\mathbb{P}}_{k}}{d\mathbb{P}}\right|_{\mathcal{F}_{t}}=\tilde{Z}_{k}(t), & \tilde{Z}_{k}(t)=\frac{S_{1}(t)^{k}}{\mathbb{E}_{\mathbb{P}}\left[S_{1}(t)^{k}\right]}=exp\left\{ k\cdot\sigma_{1}\cdot W_{1}(t)-\frac{1}{2}\cdot k^{2}\cdot\sigma_{1}^{2}\cdot t\right\} ,\end{array}\label{eq:P_k_tilde}
\end{equation}
and 
\begin{equation}
\mathbb{E}_{\mathbb{P}}\left[S_{1}(t)^{k}\right]=S_{1}(0)^{k}\cdot e^{k\cdot\mu_{1}\cdot t+\frac{1}{2}\cdot k\cdot(k-1)\cdot\sigma_{1}^{2}\cdot t}.\label{eq:ew_S_1_k_P}
\end{equation}
\end{lem}
\begin{proof}
The proof is given in \ref{pf:kth_Moment_PP}.
\end{proof}
Thus, similar to the portfolio values themselves the moments of the
payoff distributions of the performance participation strategies are
directly linked to the moments of the corresponding portfolio insurance
strategies in the discounted market. However, an additional change
of probability measure has to be conducted.

In the following, we generally derive the $k$th moments of the payoffs
of a standard OBPI and CPPI strategy. Note that the calculation of
the expected value as well as the variance was e.g. already proceeded
in \citet{BertrandPrigent2005} (expectation only) or \citet{ZagstKraus2009}.
We start with the CPPI. 
\begin{prop}[CPPI moments]
\label{thm:moments_CPPI} The $k$th moment, $k\in\mathbb{N}$, of
a standard CPPI portfolio with level of insurance $\alpha^{PI}\leq e^{r\cdot T}$
and multiplier $m$ at any time $t\in[0,T]$ is given by 
\begin{equation}
m_{k}\left(V^{CPPI}(t)\right)=\left(\alpha^{PI}\cdot V_{0}\right)^{k}\cdot\sum_{i=0}^{k}{k \choose i}\cdot e^{-k\cdot r\cdot(T-t)}\cdot\left(\frac{1-\alpha^{PI}\cdot e^{-r\cdot T}}{\alpha^{PI}\cdot e^{-r\cdot T}}\right)^{i}\cdot e^{i\cdot m\cdot\left[\mu_{S}-r+\frac{1}{2}\cdot(i-1)\cdot m\cdot\sigma_{S}^{2}\right]\cdot t}.\label{eq:m_k_CPPI}
\end{equation}
\end{prop}
\begin{proof}
The proof is given in \ref{pf:kth_Moment_CPPI}.
\end{proof}
A more sophisticated calculation leads to the following general analytic
expression for the moments of the OBPI payoff distribution. 
\begin{thm}[OBPI moments]
\label{thm:moments_OBPI}
\end{thm}
The $k$th moment, $k\in\mathbb{N}$, of the payoff of a standard
OBPI strategy with level of portfolio insurance $\alpha^{PI}\leq e^{r\cdot T}$
at maturity $T$ is given by 
\begin{align}
 & m_{k}\left(V^{OBPI}(T)\right)\label{eq:m_k_OBPI}\\
 & =\left(\alpha^{PI}\cdot V_{0}\right)^{k}+\left(\alpha^{PI}\cdot V_{0}\right)^{k}\cdot\sum_{i=1}^{k}\sum_{l=0}^{i}{k \choose i}\cdot{i \choose l}\cdot(-1)^{i-l}\cdot\left(\frac{p^{PI}\cdot S_{0}}{\alpha^{PI}\cdot V_{0}}\right)^{l}\cdot e^{l\cdot\mu_{S}\cdot T+\frac{1}{2}\cdot l\cdot(l-1)\cdot\sigma_{S}^{2}\cdot T}\cdot\Phi\left(d_{1,l}\right),\nonumber 
\end{align}
where 
\begin{equation}
d_{1,l}=\frac{\ln\left(\frac{p^{PI}S_{0}}{\alpha^{PI}\cdot V_{0}}\right)+\left[\mu_{S}+\left(l-\frac{1}{2}\right)\cdot\sigma_{S}^{2}\right]\cdot T}{\sigma_{S}\cdot\sqrt{T}}.\label{eq:d_1_l}
\end{equation}

\begin{proof}
The proof is provided in \ref{pf:kth_Moment_OBPI}.
\end{proof}
According to Lemma \ref{lem:kth_Moment_PP} the $k$th moments of
the performance participation strategies follow from the $k$th moments
of the corresponding portfolio insurance strategies with respect to
the equivalent probability measure $\tilde{\mathbb{P}}_{k}$ in the
discounted market. The corresponding asset characteristics are provided
in the following remark.
\begin{rem}
The stochastic dynamics of the index ratio $\hat{S}_{2}=S_{2}/S_{1}$
with respect to the equivalent probability measure $\tilde{\mathbb{P}}_{k}$,
$k\in\mathbb{N}$ defined in Equation (\ref{eq:P_k_tilde}) are given
by 
\begin{equation}
d\hat{S}_{2}(t)=\hat{S}_{2}(t)\hat{\mu}_{2,\tilde{k}}dt+\hat{S}_{2}(t)\hat{\sigma}_{2}d\tilde{W}_{\hat{S}_{2}}(t),\ \ \hat{S}_{2}(0)=1,\label{eq:SDE_S_2_hat_P_k_tilde}
\end{equation}
with drift 
\begin{equation}
\hat{\mu}_{2,\tilde{k}}=\hat{\mu}_{2}+k\cdot\left(\rho_{12}\cdot\sigma_{1}\cdot\sigma_{2}-\sigma_{1}^{2}\right)=\left(\mu_{2}-\mu_{1}\right)+(k-1)\cdot\left(\rho_{12}\cdot\sigma_{1}\cdot\sigma_{2}-\sigma_{1}^{2}\right),\label{eq:mu_2_hat_tilde_k}
\end{equation}
diffusion $\hat{\sigma}_{2}$ as defined in Equation (\ref{eq:sigma_2_hat})
and Wiener process 
\begin{equation}
\tilde{W}_{\hat{S}_{2},k}=\frac{\rho_{12}\cdot\sigma_{2}-\sigma_{1}}{\hat{\sigma}_{2}}\cdot\tilde{W}_{1,k}+\frac{\sqrt{1-\rho_{12}^{2}}\cdot\sigma_{2}}{\hat{\sigma}_{2}}\cdot\tilde{W}_{2,k}.\label{eq:W_S_2_hat_tilde_k}
\end{equation}
The risk-free interest rate $\hat{r}=0$ remains the same under the
change of probability measure.\end{rem}
\begin{proof}
The stochastic dynamics of $\hat{S}_{2}$ under the real-world measure
$\mathbb{P}$ are given in Remark \ref{rmk:SDE_S_2_hat}. Then, following
from the Girsanov theorem%
\footnote{See, e.g., \citep{Shreve2008}.%
} the stochastic process $\bm{\tilde{W}_{k}}(t)=\left(\tilde{W}_{1,k}(t),\tilde{W}_{2,k}(t)\right)'$,
$0\leq t\leq T$ defined by 
\begin{align}
\tilde{W}_{1,k}(t) & :=W_{1}(t)-k\cdot\sigma_{1}\cdot t,\label{eq:W_tilde_1_k}\\
\tilde{W}_{2,k}(t) & :=W_{2}(t),\label{eq:W_tilde_2_k}
\end{align}
is a two-dimensional Brownian motion under the equivalent probability
measure $\tilde{\mathbb{P}}_{k}$, $k\in\mathbb{N}$. Substituting
(\ref{eq:W_tilde_1_k}) and (\ref{eq:W_tilde_2_k}) in (\ref{eq:SDE_S_2_hat})
proves the proposition.%
\footnote{Note that following from the Lévy theorem $\tilde{W}_{\hat{S}_{2},k}$
is again a Brownian motion.%
}
\end{proof}
The moments of the CPPP and the OBPP strategy are then finally derived. 
\begin{lem}[CPPP moments]
\label{lem:moments_CPPP} The $k$th moment, $k\in\mathbb{N}$, of
a CPPP portfolio with level of performance participation $\alpha<1$
and multiplier $m$ at any time $t\in[0,T]$ is given by 
\begin{equation}
m_{k}\left(V^{CPPP}\left(t;T,V_{0},\alpha,m,S_{1},S_{2}\right)\right)=\alpha^{k}\cdot\mathbb{E}_{\mathbb{P}}\left[S_{1}(t)^{k}\right]\cdot\sum_{i=0}^{k}{k \choose i}\cdot\left(\frac{1-\alpha}{\alpha}\right)^{i}\cdot e^{i\cdot m\cdot\left[\hat{\mu}_{2,\tilde{k}}+\frac{1}{2}\cdot(i-1)\cdot m\cdot\hat{\sigma}_{2}^{2}\right]\cdot t},\label{eq:m_k_CPPP}
\end{equation}
where $\mathbb{E}_{\mathbb{P}}\left[S_{1}(t)^{k}\right]$, $\hat{\mu}_{2,\tilde{k}}$
and $\hat{\sigma}_{2}^{2}$ as defined above. \end{lem}
\begin{proof}
Substituting $\hat{V}_{0}=1$, $\hat{r}=0$, (\ref{eq:mu_2_hat_tilde_k}),
(\ref{eq:sigma_2_hat}) and (\ref{eq:m_k_CPPI}) in (\ref{eq:kth_Moment_PP_PI})
leads to the proposition.\end{proof}
\begin{lem}[OBPP moments]
\label{lem:moments_OBPP} The $k$th moment, $k\in\mathbb{N}$, of
the payoff of an OBPP strategy with level of performance participation
$\alpha<1$ at maturity $T$ is given by 
\begin{align}
 & m_{k}\left(V^{OBPP}\left(T;T,V_{0},\alpha,S_{1},S_{2}\right)\right)\label{eq:m_k_OBPP}\\
 & =\alpha^{k}\cdot\mathbb{E}_{\mathbb{P}}\left[S_{1}(t)^{k}\right]\cdot\left\{ 1+\sum_{i=1}^{k}\sum_{l=0}^{i}{k \choose i}\cdot{i \choose l}\cdot(-1)^{i-l}\cdot\left(\frac{p}{\alpha}\right)^{l}\cdot e^{l\cdot\left[\hat{\mu}_{2,\tilde{k}}+\frac{1}{2}\cdot(l-1)\cdot\hat{\sigma}_{2}^{2}\right]\cdot T}\cdot\Phi\left(\hat{d}_{k,l}\right)\right\} ,\nonumber 
\end{align}
where 
\begin{equation}
\hat{d}_{k,l}=\frac{\ln\left(\frac{p}{\alpha}\right)+\left[\hat{\mu}_{2,\tilde{k}}+\left(l-\frac{1}{2}\right)\cdot\hat{\sigma}_{2}^{2}\right]\cdot T}{\hat{\sigma}_{2}\cdot\sqrt{T}}.\label{eq:d_hat_k_l}
\end{equation}
\end{lem}
\begin{proof}
Substituting $\hat{V}_{0}=1$, (\ref{eq:mu_2_hat_tilde_k}), (\ref{eq:sigma_2_hat}),
strike $\alpha$ and (\ref{eq:m_k_OBPI}) in (\ref{eq:kth_Moment_PP_PI})
leads to the proposition.
\end{proof}
With the above general expressions for the $k$th moments we have
all the essential information to describe the entire payoff distributions
of the two performance participation (portfolio insurance) strategies.
The usually reported central moments are obtained by a final transformation.
\begin{rem}[Central moments]
\label{rmk:centralMoment_kthMoment} By applying the binomial theorem
the $k$th central moment $\mu_{k}(V)$ of a random variable $V$
follows directly from its $i$th moment $m_{i}(V)$, $i=0,...,k$
by 
\begin{equation}
\mu_{k}(V)=\mathbb{E}_{\mathbb{P}}\left[\left(V-\mathbb{E}_{\mathbb{P}}[V]\right)^{k}\right]=\sum_{i=0}^{k}{k \choose i}\cdot(-1)^{k-i}\cdot m_{i}(V)\cdot\left(m_{1}(V)^{k-i}\right).\label{eq:centralMoment_kthMoment}
\end{equation}

\end{rem}
In the sequel we will especially compare the first four (central)
moments of the payoff distributions of the two performance participation
strategies in more detail. We start with the expected strategy payoffs.
As mentioned earlier, with respect to the CPPP strategy we will analyze
various values of the multiplier $m$, among them the unique value
$m^{*}$ for which the expectations of the two strategies are equal.
Its derivation and analysis is the focus of the following section.

\subsubsection{Equality of payoff expectations}

\label{sec:equal_payoffExpectations}

The expected payoffs of the CPPP and the OBPP follow directly from
Lemma \ref{lem:moments_CPPP} and Lemma \ref{lem:moments_OBPP} as
the first moments of the resulting terminal portfolio value distributions
\begin{align}
\mu\left(V^{CPPP}\left(T;T,V_{0},\alpha,m,S_{1},S_{2}\right)\right) & =m_{1}\left(V^{CPPP}\left(T;T,V_{0},\alpha,m,S_{1},S_{2}\right)\right)\label{eq:ew_CPPP}\\
 & =\alpha\cdot V_{0}\cdot e^{\mu_{1}\cdot T}+\left(1-\alpha\right)\cdot V_{0}\cdot e^{\left[\mu_{1}+m\cdot\left(\mu_{2}-\mu_{1}\right)\right]\cdot T},\nonumber \\
\mu\left(V^{OBPP}\left(T;T,V_{0},\alpha,S_{1},S_{2}\right)\right) & =m_{1}\left(V^{OBPP}\left(T;T,V_{0},\alpha,S_{1},S_{2}\right)\right)\label{eq:ew_OBPP}\\
 & =\alpha\cdot V_{0}\cdot e^{\mu_{1}\cdot T}+V_{0}\cdot e^{\mu_{2}\cdot T}\cdot Call\left(0;T,\alpha,\hat{\mu}_{2,\tilde{1}},\hat{\sigma}_{2},p\cdot\hat{S}_{2}\right)\nonumber \\
 & =\alpha\cdot V_{0}\cdot e^{\mu_{1}\cdot T}\cdot\left(1-\Phi\left(\hat{d}_{1,0}\right)\right)+p\cdot V_{0}\cdot e^{\mu_{2}\cdot T}\cdot\Phi\left(\hat{d}_{1,1}\right),\nonumber 
\end{align}
where $\hat{\mu}_{2,\tilde{1}}=\mu_{2}-\mu_{1}$ and $\hat{d}_{k,l}$,
$l=0,1$ as defined in (\ref{eq:d_hat_k_l}).

The expected payoff of the CPPP strategy is independent of the variance-covariance
structure of the underlying risky assets. Thus, the expected return
is not affected by the additional source of risk. Moreover, it is
an exponentially growing function in the value of the multiplier $m$
if and only if the further condition $\mu_{1}<\mu_{2}$ is satisfied.
Since the multiplier controls the exposure to the active asset $S_{2}$
this is a natural expectation from the CPPP payoff sensitivity and
justifies our initial assumption made in Section \ref{sec:financialMarket}.
In contrast, an increase in the desired level of performance participation
$\alpha$ (exponentially) reduces the expected payoff (in case that
$\mu_{1}<\mu_{2}$). The enhanced participation guarantee in the reserve
asset comes at the price of a diminished cushion and thus less upside
potential stemming from a potential outperformance of the active asset
$S_{2}$.

With respect to the expected payoff of the OBPP strategy we observe
an analog sensitivity on the fraction $\alpha$. As motivated in Section
\ref{sec:OBPP} an increase in $\alpha$ is accompanied by a decrease
in the number of shares/exchange options $p$.

Equating the two expectations (\ref{eq:ew_CPPP}) and (\ref{eq:ew_OBPP})
leads to the following proposition.
\begin{lem}[Multiplier $m^{*}$]
\label{lem:mStar}For any parameterization of the financial market
(\ref{eq:SDE_risky}) and any level of performance participation $\alpha<1$
there exists a unique value $m^{*}\left(\alpha,\mu_{2}-\mu_{1},\hat{\sigma}_{2},T\right)$
of the multiplier such that 
\[
\mu\left(V^{CPPP}\left(T;T,V_{0},\alpha,m^{*},S_{1},S_{2}\right)\right)=\mu\left(V^{OBPP}\left(T;T,V_{0},\alpha,S_{1},S_{2}\right)\right),
\]
which is given by 
\begin{equation}
m^{*}\left(\alpha,\mu_{2}-\mu_{1},\hat{\sigma}_{2},T\right)=1+\frac{1}{\left(\mu_{2}-\mu_{1}\right)\cdot T}\cdot\ln\left(\frac{Call\left(0;T,\alpha,\mu_{2}-\mu_{1},\hat{\sigma}_{2},p\cdot\hat{S}_{2}\right)}{Call\left(0;T,\alpha,\hat{r},\hat{\sigma}_{2},p\cdot\hat{S}_{2}\right)}\right).\label{eq:m_star}
\end{equation}

\end{lem}
Here, $Call\left(t;T,\alpha,r_{f},\hat{\sigma}_{2},p\cdot\hat{S}_{2}\right)$
denotes the Black-Scholes value of a vanilla call option written on
$p$ shares of asset $\hat{S}_{2}$ with strike $\alpha$, risk-free
interest rate $r_{f}$, volatility $\hat{\sigma}_{2}$, evaluated
at time $t$ for maturity $T$.
\begin{proof}
Following from (\ref{eq:OBPP_OBPI}) and (\ref{eq:CPPP_CPPI}) the
problem can be reduced to the equivalent problem for the standard
portfolio insurance strategies in the discounted world and with respect
to the equivalent probability measure $\tilde{\mathbb{P}}_{1}$ ,
i.e. 
\[
\mathbb{E}_{\tilde{\mathbb{P}}_{1}}\left[\hat{V}^{CPPI}\left(T;T,\hat{V}_{0},\alpha,m^{*},\hat{r},\hat{S}_{2}\right)\right]=\mathbb{E}_{\tilde{\mathbb{P}}_{1}}\left[\hat{V}^{OBPI}\left(T;T,\hat{V}_{0},\alpha,\hat{r},\hat{S}_{2}\right)\right],
\]
where $\hat{V}_{0}=1$ and $\hat{r}=0$. The stochastic dynamics of
$\hat{S}_{2}$ with respect to $\tilde{\mathbb{P}}_{1}$ are provided
in (\ref{eq:SDE_S_2_hat_P_k_tilde}). This issue was already solved
in \citet{BertrandPrigent2005} yielding the proposed multiplier $m^{*}$.
\end{proof}
Note that since $\mu_{2}>\mu_{1}$ and thus $Call\left(0;T,\alpha,\mu_{2}-\mu_{1},\hat{\sigma}_{2},p\cdot\hat{S}_{2}\right)>Call\left(0;T,\alpha,\hat{r},\hat{\sigma}_{2},p\cdot\hat{S}_{2}\right)$
the value of the multiplier $m^{*}$ is always bigger than one. For
any value of the multiplier $m>m^{*}$ the expected payoff of the
CPPP strategy exceeds that of the OBPP strategy and vice versa.

The special multiplier $m^{*}$ is an increasing function of the investor-defined
level of performance participation $\alpha<1$. This sensitivity was
already motivated in \citet{BertrandPrigent2005} for the standard
OBPI and CPPI. Although both expected payoffs are decreasing in the
fraction $\alpha$ the CPPP is usually more sensitive to its variation.
This is mainly caused by the leveraging effect of the multiplier that
exponentially amplifies the reduction of the cushion.

As an example, Figure \ref{fig:sens_mStar_alpha} visualizes the evolution
of $m^{*}$ as a function of the level of performance participation
$\alpha$ for the standard case presented in Table \ref{tab:StandardParameterSet}.
With respect to the standard level of performance participation $\alpha=0.95$
the multiplier $m^{*}$ according to Equation (\ref{eq:m_star}) yields
$m^{*}\left(0.95,3.1\%,22.3\%,1\right)=6.90$. Note that since the
risky reserve asset features a higher drift than the risk-free asset
the associated lower excess return will usually induce higher values
of the multiplier $m^{*}$ than it is the case for the standard CPPI
strategy. 
\begin{figure}[htbp]
\centering \includegraphics[width=1\textwidth]{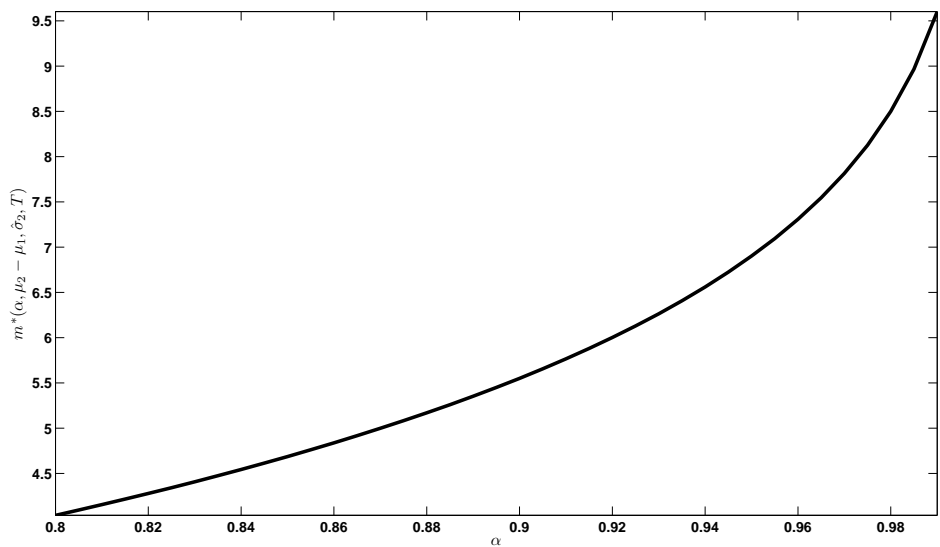}
\caption{Multiplier $m^{*}$ as a function of the investor-defined level of
performance participation $\alpha$ for the standard parameter setup
provided in Table \ref{tab:StandardParameterSet}. }

\label{fig:sens_mStar_alpha} 
\end{figure}

Furthermore, $m^{*}$ is a decreasing function in the excess drift
$\mu_{2}-\mu_{1}$. Although both expected values are increasing in
the drift difference%
\footnote{Recall that the rho of Black-Scholes standard vanilla call option
is always positive. See, e.g., \citet{Hull2009}.%
}, again the CPPP payoff reacts usually more sensitively to a change
as it is amplified by the multiplier $m$. Thus, for higher excess
drifts a smaller value of $m$ is sufficient to maintain an equivalent
level of return expectation.

In the next step we include higher (central) moments in our analysis.
Since the payoffs under consideration are non-linear a mean-variance
approach is not sufficient. This leads us to examine, besides expectation,
also standard deviation, skewness and kurtosis.

\subsubsection{Comparison of the first four moments\label{sec:first_four_Moments}}

Table \ref{tab:first_four_moments} provides the obtained values for
the expectation ($\mu$), standard deviation ($\sigma$), skewness
($\gamma_{3}$) and kurtosis ($\gamma_{4}$) of the returns of the
OBPP and the CPPP strategy in the case of the standard parameterization
provided in Table \ref{tab:StandardParameterSet}. With respect to
the CPPP strategy different values of the multiplier $m$ are analyzed
including the special value $m^{*}$. Note that for the sake of simplicity
the results are given in a return dimension instead of the usual portfolio
value dimension.

\begin{table}[htbp]
\centering %
\begin{tabular}{l||cc|ccccc}
 & OBPP  & CPPP  &  &  &  &  & \tabularnewline
 & $\alpha=0.95$  & $m^{*}=6.90$  & $m=3$  & $m=5$  & $m=6$  & $m=7$  & $m=8$\tabularnewline
\hline 
\hline 
$\mu$  & \multicolumn{2}{c}{8.10\%} & 7.34\%  & 7.72\%  & 7.91\%  & 8.12\%  & 8.33\% \tabularnewline
$\sigma$  & 12.37\%  & 19.92\%  & 5.02\%  & 9.58\%  & 13.92\%  & 20.74\%  & 31.84\%\tabularnewline
$\gamma_{3}$  & 2.3606  & 37.7639 & 1.1672 & 7.6060 & 16.9331  & 41.4930  & 118.2519 \tabularnewline
$\gamma_{4}$  & 7.4806  & 1.3912$\cdot10^{4}$  & 5.3743  & 222.9118  & 1.6409$\cdot10^{3}$  & 1.7946$\cdot10^{4}$  & 3.0765$\cdot10^{5}$ \tabularnewline
\end{tabular}\caption{Expectation, standard deviation, skewness and kurtosis of the payoff
distributions of an OBPP strategy and CPPP strategies with different
multipliers for the standard parameterization provided in Table \ref{tab:StandardParameterSet}.}

\label{tab:first_four_moments} 
\end{table}

We obtain comparable results as for standard portfolio insurance strategies
(see \citet{BertrandPrigent2005}). Both strategies generate an asymmetric
payoff profile. However, due to the significantly higher positive
skewness of the CPPP in comparison to the OBPP, it should be preferred
with respect to that criterion. Furthermore, the strategy's (excess)
kurtosis largely exceeds that of the OBPP. This feature is explained
by the outperformance of the dynamic strategy in the right tail of
the distribution where $\hat{S}_{2}(T)>>1$.

When differing from the special multiplier $m=m^{*}$ to consider
more general values of $m$, we have to distinguish two situations.
If the multiplier $m$ is higher than $m^{*}$, then the CPPP strategy
provides a higher expected payoff than the OBPP. The improved upside
potential/intensified participation in the better returning asset
$S_{2}$ is not for free and comes at the price of more risk, i.e.
an increasing standard deviation. Since the CPPP thus exceeds the
risk associated with the OBPP, none of the two strategies dominates
the other one in a mean-variance sense.

In contrast, if the multiplier $m$ takes smaller values than $m^{*}$,
then both the expected CPPP payoff as well as the strategy's standard
deviation are decreasing. Thus, the CPPP provides a smaller return
expectation than the OBPP. Furthermore, for small negative deviations
of $m$ with respect to $m^{*}$ the risk of the OBPP will still remain
less than that associated with the CPPP. Consequently the OBPP strictly
dominates the CPPP in a mean-variance sense. Nevertheless, for sufficiently
large differences $m^{*}-m$ both the expected payoff as well as the
standard deviation of the CPPP strategy take smaller values than the
OBPP ones. Hence, none of the strategies dominates the other one with
respect to the mean-variance criterion.

\citet{BertrandPrigent2005} further analyze the probability distributions
associated with the standard OBPI and the standard CPPI strategy with
a special focus on the payoff ratio $\frac{V^{OBPI}(T)}{V^{CPPI}(T)}$.
Among others they conclude that for usual values of the multiplier
the probability that the CPPI outperforms the OBPI at the terminal
date is increasing in the level of insurance. As the probability of
exercising the call option at maturity decreases with the increasing
strike, the upside potential of the OBPI strategy is significantly
reduced. Due to the special relationship between performance participation
and portfolio insurance strategies following from (\ref{eq:OBPP_OBPI})
and (\ref{eq:CPPP_CPPI}) this result remains valid for the more general
OBPP and CPPP, as 
\[
\frac{V^{OBPP}(T)}{V^{CPPP}(T)}=\frac{\hat{V}^{OBPI}(T)}{\hat{V}^{CPPI}(T)}.
\]
For further details of the analysis we refer the interested reader
to \citet{BertrandPrigent2005}.

In the following we will briefly study the dynamic properties of the
two strategies and in particular their \textquotedbl{}Greeks\textquotedbl{}.
Due to the elaborated relationship between performance participation
and standard portfolio insurance strategies the analysis can be kept
short for sensitivities where the additional source of risk is not
of direct interest.

\section{The Dynamic Behavior of OBPP and CPPP\label{sec:DynamicBehavior}}

With respect to the practical realization of the OBPP strategy, in
many situations the use of standardized traded options is not possible.
For example, the underlying(s) may be a diversified fund for which
no single option is available. Furthermore, the desired investment
period may also not coincide with the maturity of a listed option.
OTC options, on the other hand, involve several drawbacks like counterparty
risk or liquidity problems and raise the question for the fair price
of the contingent claim.

In practice, the underlying exchange options are thus usually synthesized
by dynamic replication. In the presumed Black-Scholes model (\ref{eq:SDE_risky})
the perfect hedging strategy according to the Margrabe formula (\ref{eq:value_Ex})
exists. Based on the induced dynamic hedging rule one can show that
the OBPP strategy actually represents a generalized CPPP strategy
with time-variable multiplier. Moreover, the study of the derived
multiplier allows to quantify the risk exposure associated with the
OBPP strategy.

Since in practice the number of rebalancing trades is limited due
to the associated transaction costs, the dynamic replication of the
exchange option induces hedging risks. This is the reason why we also
analyze the hedging properties (i.e. the Greeks) of both methods.
Special focus is put on the behavior of the exposures to the two risky
underlyings during the investment period.

\subsection{OBPP as generalized CPPP\label{sec:OBPP_as_CPPP}}

For the CPPP method the multiplier $m$ is the key parameter that
controls the amount invested in the active asset $S_{2}$ and the
possible leverage associated with it. Moreover, as shown in Section
\ref{sec:Moments}, it directly affects the risk-return ratio of the
resulting portfolio. Knowing about the importance of the multiplier,
\citet{BertrandPrigent2005} were able to show for the standard OBPI
strategy the existence of such an \textquotedbl{}implicit\textquotedbl{}'
parameter. This finding remains valid for the two-asset OBPP strategy.
The respective implicit multiplier for the OBPP is deduced in the
following proposition.
\begin{lem}[OBPP multiplier]
The OBPP method is equivalent to a generalized CPPP method with time-variable
multiplier given by 
\begin{equation}
m^{OBPP}\left(t;\alpha,S_{1},S_{2}\right)=\frac{p\cdot S_{2}(t)\cdot\Phi\left(d_{1}\right)}{V^{ex}\left(t;T,p\cdot S_{2},\alpha\cdot S_{1}\right)}=\frac{p\cdot\hat{S}_{2}(t)\cdot\Phi\left(d_{1}\right)}{Call\left(t;T,\alpha,\hat{r},\hat{\sigma}_{2},p\cdot\hat{S}_{2}\right)},\label{eq:m_OBPP}
\end{equation}
where $d_{1}$ as defined in (\ref{eq:d_1_ex}).\end{lem}
\begin{proof}
The proposition follows directly from \citet{BertrandPrigent2005}'s
result for standard CPPI strategies and the similarity of performance
participation and portfolio insurance according to (\ref{eq:OBPP_OBPI})
and (\ref{eq:CPPP_CPPI}) as well as (\ref{eq:Ex_discounted_Call}).
\end{proof}
In analogy to the CPPP strategy the OBPP multiplier is equal to the
ratio of the exposure to the active asset $S_{2}$ and the OBPP cushion.
The former is given by the amount invested in asset $S_{2}$ to replicate
the exchange option $V^{ex}\left(t;T,p\cdot S_{2},\alpha\cdot S_{1}\right)$
and the latter by the value of the exchange option itself. In terms
of the discounted asset universe, respectively, this is equal to the
exposure to $\hat{S}_{2}$ to replicate the corresponding call option
on the relative asset $\hat{S}_{2}$ divided by the OBPI cushion,
which is the call value.

As an example Figure \ref{fig:sens_mOBPP_S} visualizes the OBPP multiplier
as a function of the index ratio $\hat{S}_{2}$ at time $t=0.25,0.5,0.75$
for the case of the standard parameterization provided in Table \ref{tab:StandardParameterSet}.
\begin{figure}[H]
\centering \includegraphics[width=1\textwidth]{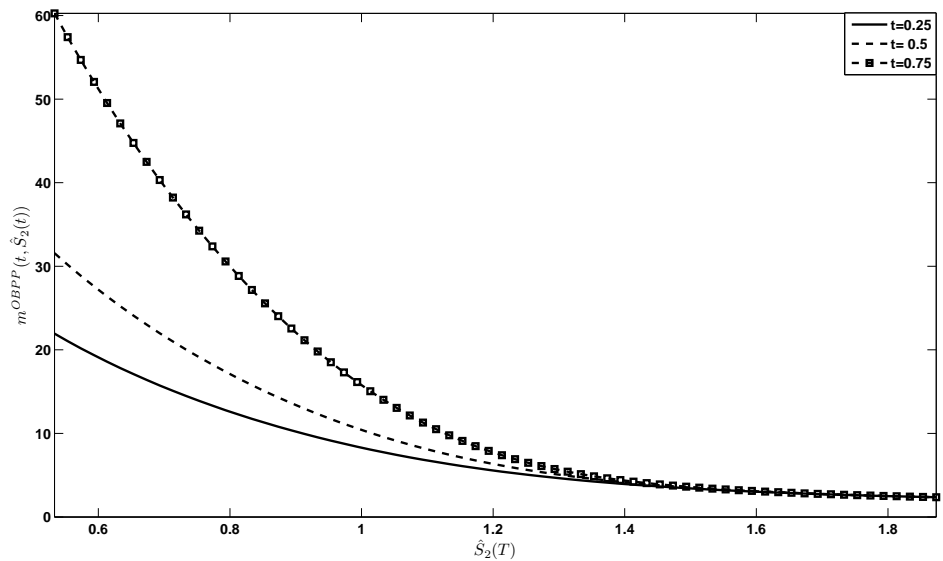}
\caption{OBPP multiplier $m^{OBPP}\left(t;\alpha,\hat{S}_{2}\right)$ as a
function of the index ratio $\hat{S}_{2}$ at time $t=0.25,0.5,0.75$
for the standard parameterization provided in Table \ref{tab:StandardParameterSet}.}

\label{fig:sens_mOBPP_S} 
\end{figure}

Hence, the OBPP multplier is at any time $t\in[0,T]$ a decreasing
function of the index ratio $\hat{S}_{2}$. Furthermore, following
from the Margrabe formula (\ref{eq:value_Ex}) it is always bigger
than $1$.

It can be easily seen that the OBPP multiplier usually takes higher
values than that of standard CPPP strategies, except for the case
when the active asset $S_{2}$ significantly outperforms the reserve
asset $S_{1}$ and the associated exchange or call option, respectively,
is thus (deeply) in the money. This is due to the small values of
the OBPP cushion, i.e. the value of the exchange option, when the
index ratio $\hat{S}_{2}$ is small.

The high values of the multiplier involve potential risk when the
market drops suddenly: it means that either the OBPP cushion is too
small or the exposure to the riskier asset is too high. In contrast,
with an increasing outperformance of the active asset $S_{2}$ over
the reserve asset $S_{1}$ the low values of the multiplier prevent
the OBPP portfolio from being overinvested in the riskier asset. Nevertheless,
a direct implication of that feature is that in the case of a sustainable
outperformance of the active asset over the entire investment period,
without any drop, the CPPP strategy will perform better due to the
higher active exposure.

To conclude the analysis of $m^{OBPP}$ Figure \ref{fig:sens_cdf_mOBPP_t}
visualizes the evolution of the cumulative distribution function of
the OBPP multiplier over time for the standard parameter set provided
in Table \ref{tab:StandardParameterSet}. 
\begin{figure}[H]
\centering \includegraphics[width=1\textwidth]{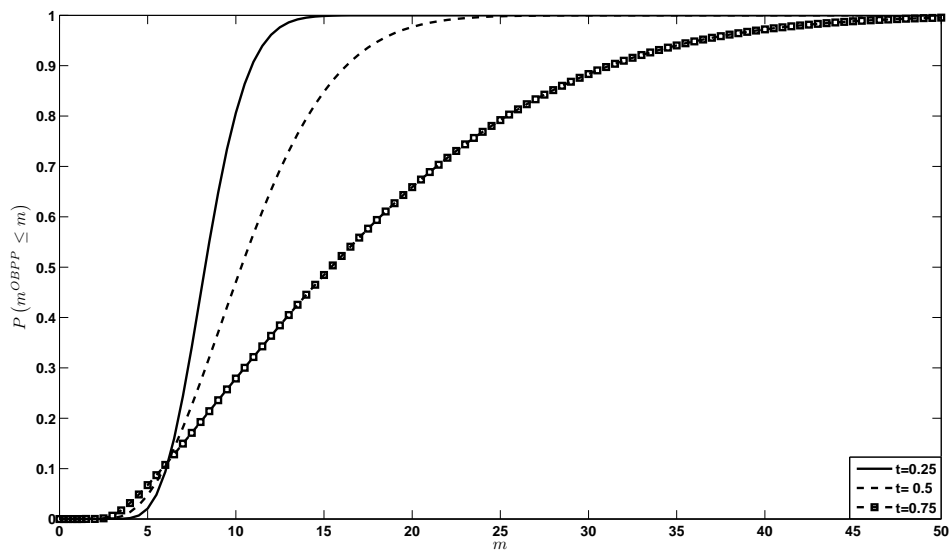}
\caption{Evolution of the cumulative distribution function of $m^{OBPP}$ for
$t=0.25,0.5,0.75$ and the standard parameter set provided in Table
\ref{tab:StandardParameterSet}.}

\label{fig:sens_cdf_mOBPP_t} 
\end{figure}

As time increases, the probability of obtaining higher values of the
OBPP multiplier increases.%
\footnote{Note that close to maturity of the exchange option (or the corresponding
call option in the discounted asset framework) the multiplier is either
infinity when the exchange (call) option is not executed, since both
the cushion and the exposure are nil. Otherwise $m^{OBPP}$ converges
to $\frac{p\cdot S_{2}}{p\cdot S_{2}-\alpha\cdot S_{1}}=\frac{p\cdot\hat{S}_{2}}{p\cdot\hat{S}_{2}-\alpha}$. %
} \citet{BertrandPrigent2005} argue that this is essentially due to
the rise of variance with time.

In the next step we will analyze the Greeks of the OBPP and the CPPP
in more detail which represent the dynamic hedging properties of the
two performance participation strategies.

\subsection{The Greeks of the performance participation strategies\label{sec:Greeks}}

We start with the delta representing one of the main concerns of a
portfolio manager.

\subsubsection{The delta\label{sec:Delta}}

Recall that the delta $\Delta$ measures the rate of change of the
portfolio value with respect to changes in the underlying asset prices.
In the case of the OBPP the sensitivities of the strategy performance
on the asset performances follow directly from the deltas of the underlying
exchange option. An analog calculation as for the derivation of the
delta of a vanilla call option yields 
\begin{align*}
\Delta_{1}^{OBPP}(t) & =\frac{\partial V^{OBPP}(t)}{\partial S_{1}(t)}=\alpha\cdot\Phi\left(-d_{2}\right)=\alpha\cdot\left(1-\Phi\left(d_{2}\right)\right),\\
\Delta_{2}^{OBPP}(t) & =\frac{\partial V^{OBPP}(t)}{\partial S_{2}(t)}=p\cdot\Phi\left(d_{1}\right).
\end{align*}
The deltas of the OBPP strategy are always positive, smaller (or equal)
to one and tend to $0$ ($1$) and $1$ ($0$) as the investment horizon
approaches its end and $p\cdot\hat{S}_{2}>\alpha$ ($p\cdot\hat{S}_{2}<\alpha$),
respectively. In contrast, if the option is at-the-money, i.e. $p\cdot\hat{S}_{2}=\alpha$,
the deltas converge to $\frac{\alpha}{2}$ and $\frac{p}{2}$, respectively,
as time to maturity approaches zero. This limiting behavior is easy
to understand from the payoff structure of the exchange option: as
the time to maturity tends to zero the OBPP portfolio is mainly invested
in the better performing asset. However, no leveraging is included.

For the CPPP strategy the deltas follow by partial differentiation
of the portfolio value (\ref{eq:value_CPPP}) as 
\begin{align}
\Delta_{1}^{CPPP}(t) & =\frac{\partial V^{CPPP}(t)}{\partial S_{1}(t)}=\alpha+(1-m)\cdot\beta_{CPPP}\left(t;\alpha,m,\hat{\sigma}_{2}\right)\cdot\left(\frac{S_{2}(t)}{S_{1}(t)}\right)^{m},\label{eq:delta_1_CPPP}\\
\Delta_{2}^{CPPP}(t) & =\frac{\partial V^{CPPP}(t)}{\partial S_{2}(t)}=m\cdot\beta_{CPPP}\left(t;\alpha,m,\hat{\sigma}_{2}\right)\cdot\left(\frac{S_{2}(t)}{S_{1}(t)}\right)^{m-1}.\label{eq:delta_2_CPPP}
\end{align}
Whereas similar to the OBPP strategy $\Delta_{2}^{CPPP}$ is always
positive, $\Delta_{1}^{CPPP}$ is always smaller than the level of
performance participation $\alpha<1$ and especially for high values
of the multiplier $m>1$ and a strong outperformance of $S_{2}$ with
respect to $S_{1}$ likely to take (highly) negative values. Simultaneously,
the $\Delta_{2}^{CPPP}$ is then very likely to take significantly
higher values than one. The observed sensitivity is due to the convex
leveraging feature of the CPPP strategy in the active asset. The strategy
value as well as $\Delta_{2}^{CPPP}$ become the more convex in $S_{2}$
the higher the value of the multiplier $m$. A high stochastic floor
in terms of $S_{1}$ thus significantly reduces the cushion and the
overall performance of the strategy. Furthermore, since $\beta_{CPPP}\left(t;\alpha,m,\hat{\sigma}_{2}\right)$
is a decreasing function in time $t$, $\Delta_{1}^{CPPP}$ is increasing
and $\Delta_{2}^{CPPP}$ is decreasing with time.

Figure \ref{fig:sens_deltaS1} and \ref{fig:sens_deltaS2} show the
evolution of the deltas of the two performance participation strategies
with respect to the reserve asset $S_{1}$ and the active asset $S_{2}$
as functions of the index ratio $\hat{S}_{2}$ at time $t=0.75$.
The standard parameter set given in Table \ref{tab:StandardParameterSet}
is applied. With respect to the CPPP strategy different values of
the multiplier are analyzed. 
\begin{figure}[htbp]
\centering \includegraphics[width=1\textwidth]{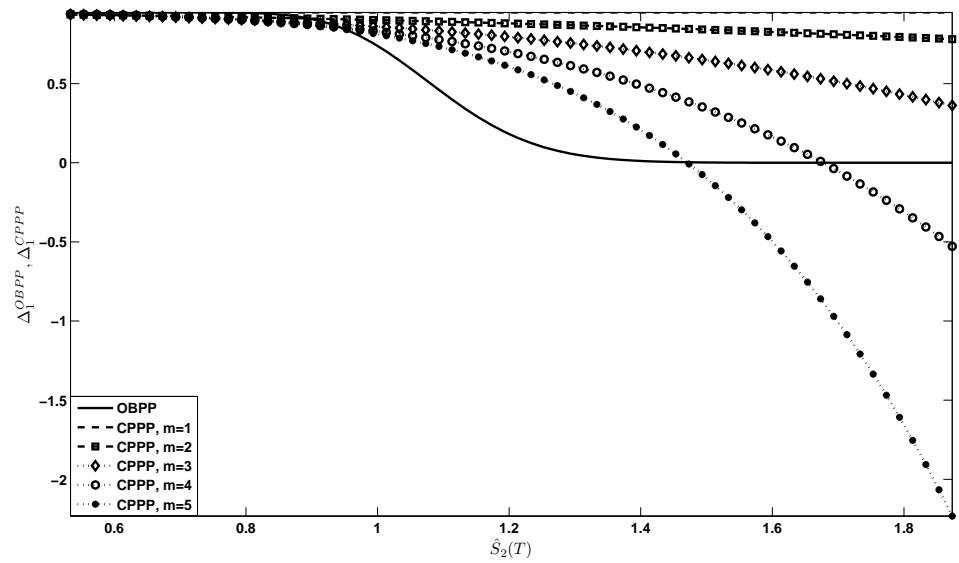}
\caption{OBPP and CPPP delta with respect to the reserve asset $S_{1}$ as
a function of the index ratio $\hat{S}_{2}$ at time $t=0.75$. The
standard parameter set provided in Table \ref{tab:StandardParameterSet}
is applied. With respect to the CPPP the multipliers $m=1,2,3,4,5$
are analyzed.}

\label{fig:sens_deltaS1} 
\end{figure}

\begin{figure}[htbp]
\centering \includegraphics[width=1\textwidth]{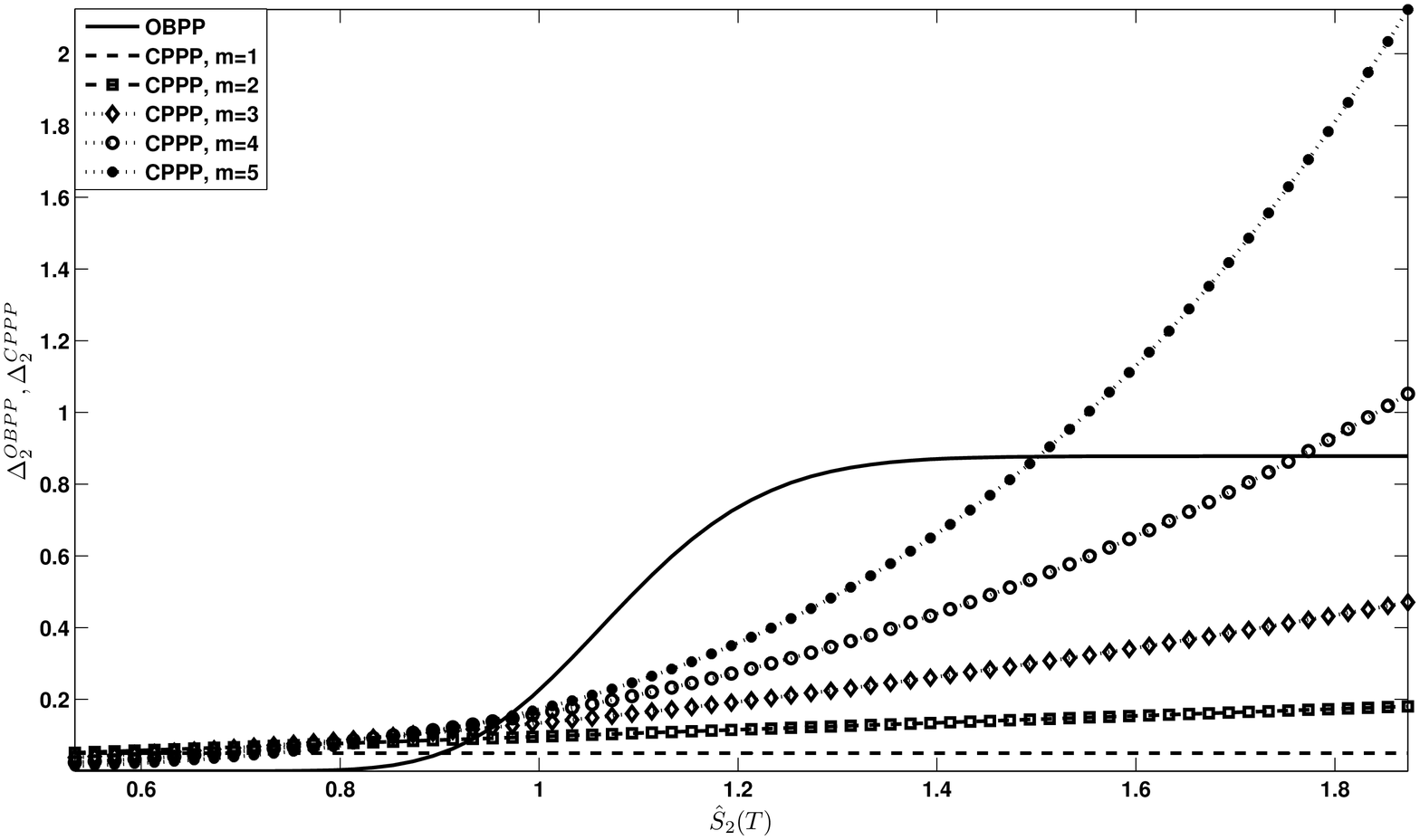}
\caption{OBPP and CPPP delta with respect to the active asset $S_{2}$ as a
function of the index ratio $\hat{S}_{2}$ at time $t=0.75$. The
standard parameter set provided in Table \ref{tab:StandardParameterSet}
is applied. With respect to the CPPP the multipliers $m=1,2,3,4,5$
are analyzed.}

\label{fig:sens_deltaS2} 
\end{figure}

Notice that for moderate $\hat{S}_{2}>1$, i.e. when the exchange
(or the discounted call) option is in-the-money, the OBPP delta with
respect to the active asset $S_{2}$ (with respect to the reserve
asset $S_{1}$) is greater (smaller) than that of the CPPP one. This
induces a higher risk of the option-based strategy in the case of
a sudden market downturn of $S_{2}$ relative to $S_{1}$. In contrast,
if the active asset significantly outperforms the reserve asset, i.e.
when the option is deeply in-the-money, then the CPPP fall exceeds
the OBPP one.

\citet{BertrandPrigent2005} analyze the delta of standard portfolio
insurance strategies in more detail. In fact, due to the relationship
between performance participation and portfolio insurance strategies
according to (\ref{eq:OBPP_OBPI}) and (\ref{eq:CPPP_CPPI}) the deltas
of the the OBPP and the CPPP with respect to the active asset $S_{2}$
are actually equal to the deltas of the corresponding standard portfolio
insurance strategies with respect to the discounted underlying $\hat{S}_{2}$
\[
\Delta_{2}^{PP}(t)=\frac{\partial V^{PP}(t)}{\partial S_{2}(t)}=\frac{\partial S_{1}(t)\cdot\hat{V}^{PI}(t)}{\partial S_{2}(t)}=\frac{\partial\hat{V}^{PI}(t)}{\partial\hat{S}_{2}(t)}=\hat{\Delta}^{PI}.
\]
Thus, all results derived by \citet{BertrandPrigent2005} remain valid
for the generalized performance participation strategies in terms
of the index ratio $\hat{S}_{2}$ instead of the former risky asset
$S$. Among others they show that in probability the CPPP is significantly
less sensitive to the ratio $\hat{S}_{2}$ than the OBPP.

Next we will analyze the gamma of the OBPP and the CPPP that measures
the convexity of the portfolio values in the underlying asset prices.
It is important to ensure an effective delta hedge.

\subsubsection{The gamma\label{sec:Gamma}}

Recall that the gamma $\Gamma$ measures the rate of change in the
delta with respect to changes in the underlying prices. In the case
of the OBPP it follows directly from the gamma of the underlying exchange
option. An analog calculation to the derivation of a call option's
gamma yields 
\begin{align*}
\Gamma_{1}^{OBPP}(t) & =\frac{\partial^{2}V^{OBPP}(t)}{\partial S_{1}(t)^{2}}=\alpha\cdot\frac{\varphi\left(-d_{2}\right)}{S_{1}(t)\cdot\hat{\sigma}_{2}\cdot\sqrt{T-t}}=\alpha\cdot\frac{\varphi\left(d_{2}\right)}{S_{1}(t)\cdot\hat{\sigma}_{2}\cdot\sqrt{T-t}},\\
\Gamma_{2}^{OBPP}(t) & =\frac{\partial^{2}V^{OBPP}(t)}{\partial S_{2}(t)^{2}}=p\cdot\frac{\varphi\left(d_{1}\right)}{S_{2}(t)\cdot\hat{\sigma}_{2}\cdot\sqrt{T-t}}.
\end{align*}
Hence, the OBPP gammas are always positive and converge - except for
the at-the-money-case - to zero as the option approaches its maturity.
When the option is at-the-money, i.e. $p\cdot\hat{S}_{2}=\alpha$,
then $d_{i}$, $i=1,2$ converges to zero and $\varphi\left(d_{i}\right)$
takes the value $\frac{1}{\sqrt{2\pi}}$. Overall, the gammas diverge
to infinity in this case.

Analogously, the gammas of the CPPP strategy follow by partial differentiation
of the derived deltas (\ref{eq:delta_1_CPPP}) and (\ref{eq:delta_2_CPPP})
as 
\begin{align*}
\Gamma_{1}^{CPPP}(t) & =\frac{\partial^{2}V^{CPPP}(t)}{\partial S_{1}(t)^{2}}=m\cdot(m-1)\cdot\beta_{CPPP}\left(t;\alpha,m,\hat{\sigma}_{2}\right)\cdot S_{1}(t)^{-1}\cdot\left(\frac{S_{2}(t)}{S_{1}(t)}\right)^{m},\\
\Gamma_{2}^{CPPP}(t) & =\frac{\partial^{2}V^{CPPP}(t)}{\partial S_{2}(t)^{2}}=m\cdot(m-1)\cdot\beta_{CPPP}\left(t;\alpha,m,\hat{\sigma}_{2}\right)\cdot S_{2}(t)^{-1}\cdot\left(\frac{S_{2}(t)}{S_{1}(t)}\right)^{m-1}.
\end{align*}
Thus, the CPPP gammas are always positive justifying the convex leveraging
feature of the strategy.

Figure \ref{fig:sens_gammaS1} and \ref{fig:sens_gammaS2} visualize
the OBPP and the CPPP gammas as a function of the index ratio $\hat{S}_{2}$
at time $t=0.75$. The standard parameterization provided in Table
\ref{tab:StandardParameterSet} is applied. With respect to the CPPP
different values of the multiplier are analyzed. Note that the common
factor $S_{1}(t)^{-1}$ or $S_{2}(t)^{-1}$ with respect to $\Gamma_{1}^{OBPP,CPPP}$
or $\Gamma_{2}^{OBPP,CPPP}$, respectively, has been omitted. 
\begin{figure}[htbp]
\centering \includegraphics[width=1\textwidth]{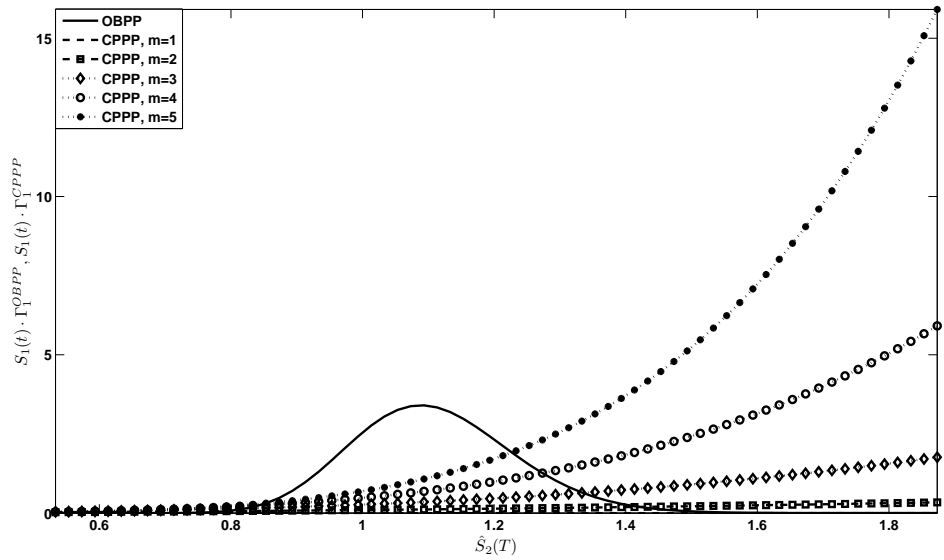}
\caption{OBPP and CPPP gamma with respect to the reserve asset $S_{1}$ as
a function of the index ratio $\hat{S}_{2}$ for $t=0.75$. The standard
parameterization provided in Table \ref{tab:StandardParameterSet}
is applied. With respect to the CPPP the multipliers $m=1,2,3,4,5$
are analyzed.}

\label{fig:sens_gammaS1} 
\end{figure}

\begin{figure}[htbp]
\centering \includegraphics[width=1\textwidth]{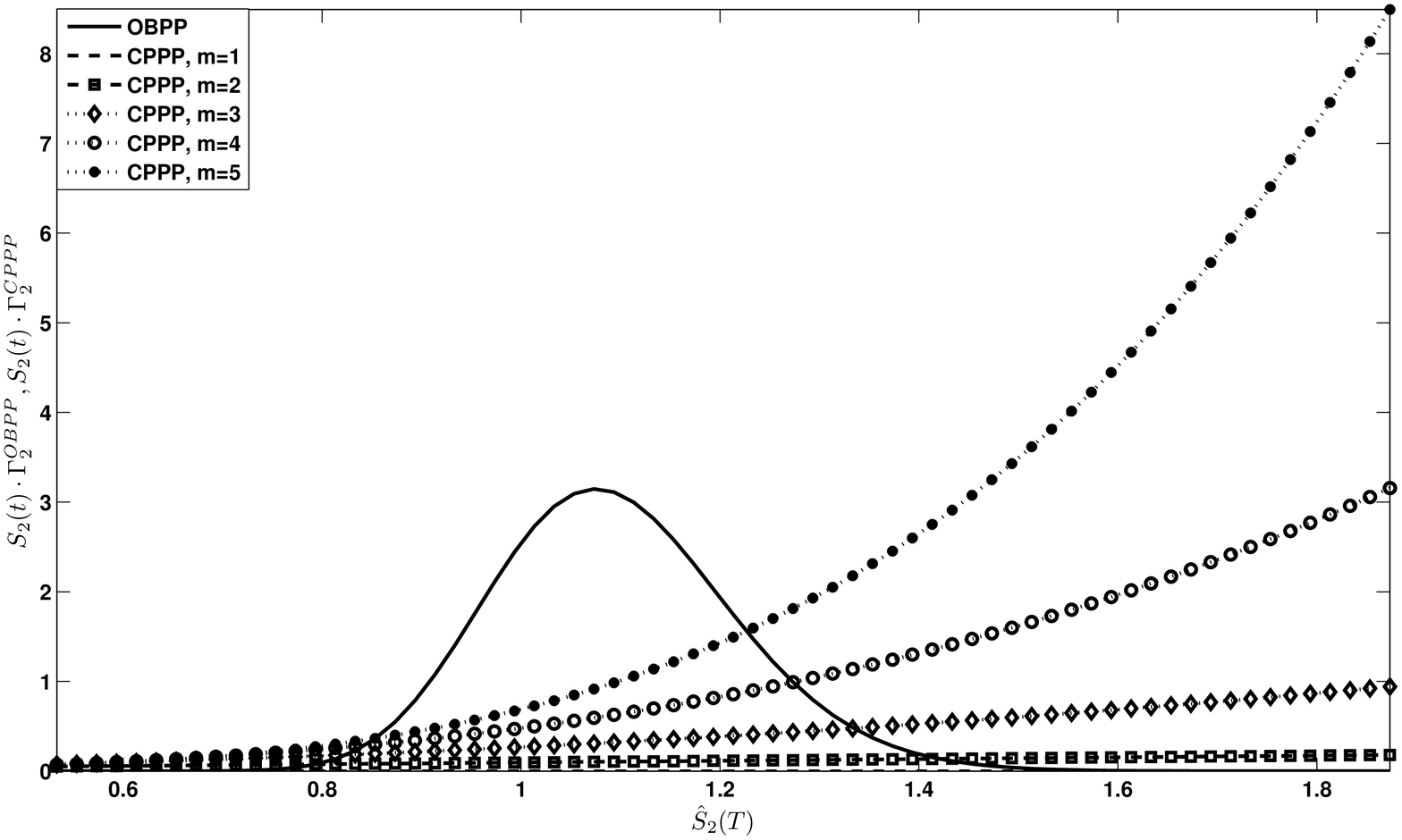}
\caption{OBPP and CPPP gamma with respect to the active asset $S_{2}$ as a
function of the index ratio $\hat{S}_{2}$ for $t=0.75$. The standard
parameterization provided in Table \ref{tab:StandardParameterSet}
is applied. With respect to the CPPP the multipliers $m=1,2,3,4,5$
are analyzed.}

\label{fig:sens_gammaS2} 
\end{figure}

With respect to the CPPP strategy the gammas are especially important
for very high values of the index ratio $\hat{S}_{2}$. Nevertheless,
for the most probable realizations of $\hat{S}_{2}$ around one, and
especially when the exchange option is in-the-money, the CPPP gamma
is smaller than the OBPP one for a large range of values. Furthermore,
since $\beta_{CPPP}\left(t;\alpha,m,\hat{\sigma}_{2}\right)$ is monotonically
decreasing in time $t$, the CPPP gammas are decreasing when approaching
the end of the investment horizon; yet, they never reach zero. In
contrast, with the time to maturity approaching zero the OBPP gammas
will converge to zero for an in- or out-of-the money call and diverge
if $p\cdot\hat{S}_{2}=\alpha$.

To conclude the section we briefly look at the portfolio vegas of
the two performance participation strategies under consideration.

\subsubsection{The vega\label{sec:Vega}}

The vega measures the strategy's sensitivity on the volatility(ies)
of the underlying assets. Since both the CPPP and the OBPP portfolio
value only depend on the aggregated volatility $\hat{\sigma}_{2}$
of the index ratio $\hat{S}_{2}$ we restrict our analysis of the
vega to this special diffusion. We recall that $\hat{\sigma}_{2}$
is a decreasing function in the asset correlation $\rho_{12}$. Furthermore,
using the assumption $\sigma_{1}<\sigma_{2}$ it is increasing in
$\sigma_{2}$; with respect to $\sigma_{1}$ the sensitivity is ambigous.

Due to the special relationship between performance participation
and portfolio insurance strategies according to (\ref{eq:OBPP_OBPI})
and (\ref{eq:CPPP_CPPI}) the vegas of the two strategies follow directly
from the vegas of the underlying portfolio insurance strategies in
the discounted market. More precisely, we obtain%
\footnote{See, e.g., \citet{Hull2009} for the call vega and \citet{BertrandPrigent2005}
for the vega of the CPPI strategy. Note that similar to \citet{BertrandPrigent2005}
the effect of the volatility on the initial price of the call/exchange
option is not taken into account since it only depends on the expected
volatility and not on the actual one.%
} 
\begin{align*}
\nu^{OBPP}(t) & =\frac{\partial V^{OBPP}(t)}{\partial\hat{\sigma}_{2}}=p\cdot S_{1}(t)\cdot\hat{S}_{2}(t)\cdot\varphi\left(d_{1}\right)\cdot\sqrt{T-t}=p\cdot S_{2}(t)\cdot\varphi\left(d_{1}\right)\cdot\sqrt{T-t},\\
\nu^{CPPP}(t) & =\frac{\partial V^{CPPP}(t)}{\partial\hat{\sigma}_{2}}=S_{1}(t)\cdot m\cdot(1-m)\cdot\hat{\sigma}_{2}\cdot t\cdot\left(\hat{V}^{CPPI}(t)-\hat{F}^{CPPI}(t)\right)\\
 & =-m\cdot(m-1)\cdot\hat{\sigma}_{2}\cdot t\cdot C(t).
\end{align*}
Whereas the vega of the OBPP strategy is always positive, the vega
of the CPPP strategy takes negative values for $m>1$. Thus, an increase
in the volatility of the index ratio $\hat{S}_{2}$ reduces the actual
CPPP portfolio value. The extent of the decrease is the larger the
higher the value of the multiplier $m$ and the longer the elapsed
investment time $t$ since inception. If $m=1$ the value of the CPPP
strategy is independent on a change in the tracking error $\hat{\sigma}_{2}$.
In contrast, the vega of the exchange/call option is positive, since
an increase in the volatility increases the probability that the option
will be executed at its expiry and thus its price. The closer the
investment horizon approaches maturity the smaller is the impact of
a change in $\hat{\sigma}_{2}$ since $\lim_{t\rightarrow T}\nu^{OBPP}(t)=0$.

As an example Figure \ref{fig:sens_vega} visualizes the discounted
OBPP and CPPP vega as a function of the index ratio $\hat{S}_{2}$
at time $t=0.75$. The standard model parameterization provided in
Table \ref{tab:StandardParameterSet} is applied. With respect to
the CPPP different values of the multiplier are analyzed. 
\begin{figure}[htbp]
\centering \includegraphics[width=1\textwidth]{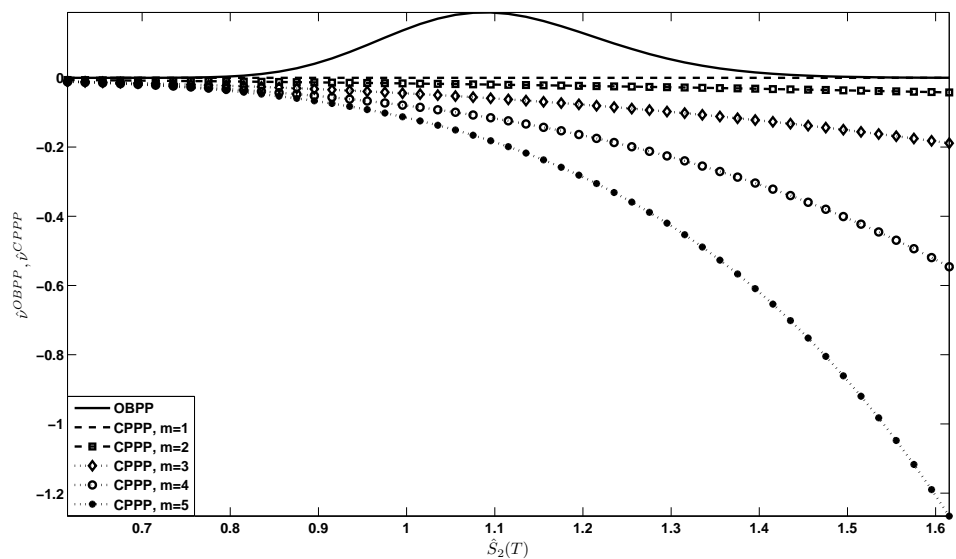}
\caption{Discounted OBPP and CPPP vega as functions of the index ratio $\hat{S}_{2}$
for $t=0.75$. The standard parameter set provided in Table \ref{tab:StandardParameterSet}
is applied. With respect to the CPPP different values of the multiplier
are analyzed.}

\label{fig:sens_vega} 
\end{figure}

Table \ref{tab:sensitivitiesGreeks} summarizes the described sensitivities
of the strategy Greeks with respect to the magnitude of the index
ratio $\hat{S}_{2}$. 
\begin{table}[htbp]
\centering %
\begin{tabular}{l||ccc|rcl|ccc}
 & \multicolumn{3}{c|}{$\hat{S}_{2}<<1$} & \multicolumn{3}{c|}{$\hat{S}_{2}\approx1$} & \multicolumn{3}{c}{$\hat{S}_{2}>>1$}\tabularnewline
 & OBPP  &  & CPPP  & OBPP  &  & CPPP  & OBPP  &  & CPPP \tabularnewline
\hline 
\hline 
$\Delta_{1}$  &  & $\approx$  &  & neg.  & $<$  & neg.  & neg.  & $>>$  & neg. \tabularnewline
$\Delta_{2}$  &  & $\approx$  &  & pos.  & $>$  & pos.  & pos.  & $<<$  & pos. \tabularnewline
\hline 
$\Gamma_{1}$  &  & $\approx$  &  & pos.  & $>$  & pos.  & pos.  & $<<$  & pos. \tabularnewline
$\Gamma_{2}$  &  & $\approx$  &  & pos.  & $>$  & pos.  & pos.  & $<<$  & pos. \tabularnewline
\hline 
$\nu$  &  & $\approx$  &  & pos.  & $>$  & neg.  & pos.  & $>>$  & neg. \tabularnewline
\end{tabular}\caption{Sensitivities of the OBPP and CPPP Greeks with respect to the index
ratio $\hat{S}_{2}$.}

\label{tab:sensitivitiesGreeks} 
\end{table}

Overall, the two portfolios react in a similar way on a change in
the value of the index ratio $\hat{S}_{2}$ or the associated volatility
$\hat{\sigma}_{2}$ when the active asset is significantly outperformed
by the reserve asset. Thus, neither of the two strategies can be preferred.
In contrast, for exceptionally high values of $\hat{S}_{2}$ the CPPP
strategy requires significantly more hedging effort than the OBPP
as it is extremely sensitive on changes in the value of the active
asset or the volatility of the index ratio. This sensitivity is even
amplified with higher values of the multiplier $m$. Nevertheless,
for the most probable realizations of the asset ratio $\hat{S}_{2}$
around one, the OBPP appears riskier with respect to a sudden market
drop.

To conclude our analysis of the OBPP and the CPPP strategy we summarize
the main results and give some concluding remarks.

\section{Conclusion}

In this paper we have introduced the class of performance participation
strategies that guarantee a minimum return in terms of a percentage
of a stochastic benchmark while keeping the potential for profits
from the outperformance of a second, riskier asset. The new strategy
class thus represents a generalization of the well-established portfolio
insurance methods where the provided guarantee is not deterministic
anymore but subject to systematic risk instead.

Moreover, with the OBPP and the CPPP we have presented a static as
well as a dynamic example of a performance participation method that
are closely related to the well-known OBPI and CPPI strategy. In fact,
we have shown that the two strategy classes can be transformed into
each other by discounting with the reserve asset as numéraire. Based
on that very important relationship we were able to derive general
analytic expressions not only for all of the moments of the payoff
distributions of the performance participation strategies but also
for the standard OBPI and CPPI method.

In the subsequent analysis we have compared the OBPP and the CPPP
with respect to various criteria, including the payoff distributions
as well as the dynamic behavior. We conclude that neither of the two
strategies generally dominates the other one. This comes from the
non-linearity of the payoff functions. Nevertheless, as the investor-defined
level of performance participation increases the CPPP strategy seems
to be more relevant than the OBPP one. Since the probability of exercising
the exchange option at maturity decreases with the desired participation
level, the upside potential of the OBPP is significantly reduced.

A concluding analysis of the dynamic behavior of the two strategies
showed that the (synthesized) OBPP can actually be considered as a
generalized CPPP strategy with time-variable multiplier. Although
the OBPP payoff exceeds the CPPP one when the associated exchange
option is around or slightly in the money it is more sensitive to
drops in the index ratio $\hat{S}_{2}$ and to transaction costs.
Furthermore, since the OBPP inherent multiplier represents a decreasing
function of $\hat{S}_{2}$, especially in the case of constantly rising
markets the CPPP is likely to outperform the OBPP.

So far we have restricted our analysis of the OBPP and the CPPP to
the comparison of the moments of the payoff distributions as well
as the dynamic behavior. However, a detailed analysis should also
include investor-specific utility functions and criteria of stochastic
dominance which was analyzed for the standard OBPI and CPPI in \citep{ZagstKraus2009}.
This will be the subject of further research.

\appendix

\section{Moments (Proof of Lemma \ref{lem:kth_Moment_PP})}

\label{pf:kth_Moment_PP} Following from (\ref{eq:OBPP_OBPI}) and
(\ref{eq:CPPP_CPPI}) the current value $V^{PP}(t)$ of the performance
participation strategy PP is given by 
\[
V^{PP}(t)=S_{1}(t)\cdot\hat{V}^{PI}(t),
\]
where $\hat{V}^{PI}(t)$ is the current value of the associated portfolio
insurance strategy in the discounted market with risk-free interest
rate $\hat{r}=0$ and risky asset $\hat{S}_{2}=S_{2}/S_{1}$. Thus,
the $k$th moment, $k\in\mathbb{N}$ yields 
\[
m_{k}\left(V^{PP}(t)\right)=\mathbb{E}_{\mathbb{P}}\left[S_{1}(t)^{k}\cdot\hat{V}^{PI}(t)^{k}\right],
\]
where $\mathbb{E}_{\mathbb{P}}$ denotes the (unconditional) expectation
with respect to the real-world measure $\mathbb{P}$. Define for $k\in\mathbb{N}$
the equivalent probability measure $\tilde{\mathbb{P}}_{k}$ via the
Radon-Nikodym derivative%
\footnote{See, e.g., \citet{Shreve2008}.%
} 
\[
\begin{array}{cc}
\left.\frac{d\tilde{\mathbb{P}}_{k}}{d\mathbb{P}}\right|_{\mathcal{F}_{t}}=\tilde{Z}_{k}(t), & \tilde{Z}_{k}(t)=\frac{S_{1}(t)^{k}}{\mathbb{E}_{\mathbb{P}}\left[S_{1}(t)^{k}\right]}=exp\left\{ k\cdot\sigma_{1}\cdot W_{1}(t)-\frac{1}{2}\cdot k^{2}\cdot\sigma_{1}^{2}\cdot t\right\} ,\end{array}
\]
with 
\[
\mathbb{E}_{\mathbb{P}}\left[S_{1}(t)^{k}\right]=S_{1}(0)^{k}\cdot e^{k\cdot\mu_{1}\cdot t+\frac{1}{2}\cdot k\cdot(k-1)\cdot\sigma_{1}^{2}\cdot t}.
\]
Then, applying the Bayes rule%
\footnote{See, e.g., \citet{Shreve2008}.%
} and substituting the explicit asset price $S_{1}(t)$ leads to 
\begin{align*}
m_{k}\left(V^{PP}(t)\right) & =\tilde{\mathbb{E}}_{k}\left[\tilde{Z}_{k}(t)^{-1}\cdot S_{1}(t)^{k}\cdot\hat{V}^{PI}(t)^{k}\right]\\
 & =\mathbb{E}_{\mathbb{P}}\left[S_{1}(t)^{k}\right]\cdot\tilde{m}_{k}\left(\hat{V}^{PI}(t)\right),
\end{align*}
where $\tilde{m}_{k}\left(\hat{V}^{PI}(t)\right)=\tilde{\mathbb{E}}_{k}\left[\hat{V}^{PI}(t)^{k}\right]$,
denotes the $k$th moment of the associated portfolio insurance strategy
with respect to the equivalent probability measure $\tilde{\mathbb{P}}_{k}$.

\section{Moments of the CPPI strategy (Proof of Theorem \ \ref{thm:moments_CPPI})}

\label{pf:kth_Moment_CPPI}The portfolio value of the CPPI strategy
(\ref{eq:value_CPPI}) at time $t\in[0,T]$ is given by 
\[
V^{CPPI}(t)=F^{CPPI}(t)+C^{CPPI}(t),
\]
with deterministic floor 
\[
F^{CPPI}(t)=\alpha^{PI}\cdot V_{0}\cdot e^{-r\cdot(T-t)},
\]
and lognormally distributed cushion process 
\[
C^{CPPI}(t)=\beta_{CPPI}\left(t;\alpha^{PI},m,r,\sigma_{S}\right)\cdot V_{0}\cdot e^{r\cdot t}\cdot\left(\frac{S(t)}{V_{0}\cdot e^{r\cdot t}}\right)^{m},
\]
where%
\footnote{See, e.g., \citet{BertrandPrigent2005} or \citep{ZagstKraus2009}.%
} 
\begin{align}
dC^{CPPI}(t) & =C^{CPPI}\mu_{C^{CPPI}}dt+C^{CPPI}\sigma_{C^{CPPI}}dW_{S}(t), & C_{0}^{CPPI}=V_{0}\cdot\left(1-\alpha^{PI}\cdot e^{-r\cdot T}\right),\label{eq:SDE_C_CPPI}\\
\mu_{C^{CPPI}} & =r+m\cdot\left(\mu_{S}-r\right),\label{eq:mu_C_CPPI}\\
\sigma_{C^{CPPI}} & =m\cdot\sigma_{S}.\label{eq:sigma_C_CPPI}
\end{align}
Hence, by applying the binomial theorem the $k$th moment can be decomposed
as 
\[
m_{k}\left(V^{CPPI}(t)\right)=\sum_{i=0}^{k}{k \choose i}\cdot m_{i}\left(C^{CPPI}(t)\right)\cdot\left(F^{CPPI}(t)\right)^{k-i},
\]
where $m_{i}\left(C^{CPPI}(t)\right)=\mathbb{E}_{\mathbb{P}}\left[C^{CPPI}(t)^{i}\right]$
and $\mathbb{E}_{\mathbb{P}}$ denotes the (unconditional) expectation
with respect to the real-world measure $\mathbb{P}$. Together with
\[
m_{i}\left(C^{CPPI}(t)\right)=\left(C_{0}^{CPPI}\right)^{i}\cdot\exp\left\{ i\cdot\left(\mu_{C^{CPPI}}-\frac{1}{2}\cdot\sigma_{C^{CPPI}}^{2}\right)\cdot t+\frac{1}{2}\cdot i^{2}\cdot\sigma_{C^{CPPI}}^{2}\cdot t\right\} ,
\]
as well as (\ref{eq:mu_C_CPPI}) and (\ref{eq:sigma_C_CPPI}) this
finally leads to 
\begin{align*}
 & m_{k}\left(V^{CPPI}(t)\right)\\
 & =\sum_{i=0}^{k}{k \choose i}\cdot\left(C_{0}^{CPPI}\right)^{i}\cdot\exp\left\{ i\cdot\mu_{C}\cdot t+\frac{1}{2}\cdot i\cdot(i-1)\cdot\sigma_{C}^{2}\cdot t\right\} \cdot\left(\alpha^{PI}\right)^{k-i}\cdot V_{0}^{k-i}\cdot e^{-r\cdot(k-i)\cdot(T-t)}\\
 & =\left(\alpha^{PI}\cdot V_{0}\right)^{k}\cdot\sum_{i=0}^{k}{k \choose i}\cdot e^{-k\cdot r\cdot(T-t)}\cdot\left(\frac{1-\alpha^{PI}\cdot e^{-r\cdot T}}{\alpha^{PI}\cdot e^{-r\cdot T}}\right)^{i}\cdot e^{i\cdot m\cdot\left[\left(\mu_{S}-r\right)+\frac{1}{2}\cdot(i-1)\cdot m\cdot\sigma_{S}^{2}\right]\cdot t}.
\end{align*}

\section{Moments of the OBPI strategy (Proof of Theorem \ \ref{thm:moments_OBPI})}

\label{pf:kth_Moment_OBPI}The payoff of the OBPI strategy (\ref{eq:value_OBPI})
at maturity $T$ is given by 
\[
V^{OBPI}(T)=\alpha^{PI}\cdot V_{0}+\left(p^{PI}\cdot S(T)-\alpha^{PI}\cdot V_{0}\right)^{+}.
\]
Hence, by applying the binomial theorem the $k$th moment can be decomposed
as 
\begin{equation}
m_{k}\left(V^{OBPI}(T)\right)=\left(\alpha^{PI}\cdot V_{0}\right)^{k}+\sum_{i=1}^{k}{{k \choose i}\cdot\left(\alpha^{PI}\cdot V_{0}\right)^{k-i}\cdot UPM_{i}\left(p^{PI}\cdot S(T),\alpha^{PI}\cdot V_{0}\right)},\label{eq:m_k_OBPI_UPM}
\end{equation}
where $UPM_{i}\left(p^{PI}\cdot S(T),\alpha^{PI}\cdot V_{0}\right)=\mathbb{E}_{\mathbb{P}}\left[\left(\left(p^{PI}\cdot S(T)-\alpha^{PI}\cdot V_{0}\right)^{+}\right)^{i}\right]$,
for $i\in\mathbb{N}$ denotes the $i$th upper partial moment of $p^{PI}$
shares of the terminal asset price $S(T)$ with respect to the benchmark
$\alpha^{PI}\cdot V_{0}$. $\mathbb{E}_{\mathbb{P}}$ is the (unconditional)
expectation with respect to the real-world measure $\mathbb{P}$.

Reapplication of the binomial theorem reduces the calculation of $UPM_{i}\left(p^{PI}\cdot S(T),\alpha^{PI}\cdot V_{0}\right)$
to 
\begin{align}
 & UPM_{i}\left(p^{PI}\cdot S(T),\alpha^{PI}\cdot V_{0}\right)\label{eq:UPM_power}\\
 & =\mathbb{E}_{\mathbb{P}}\left[\left(p^{PI}\cdot S(T)\cdot\mathds{1}_{p^{PI}\cdot S(T)\geq\alpha^{PI}\cdot V_{0}}-\alpha^{PI}\cdot V_{0}\cdot\mathds{1}_{p^{PI}\cdot S(T)\geq\alpha^{PI}\cdot V_{0}}\right)^{i}\right]\\
 & =\sum_{l=0}^{i}{i \choose l}\cdot(-1)^{i-l}\cdot\left(\alpha^{PI}\cdot V_{0}\right)^{i-l}\cdot\mathbb{E}_{\mathbb{P}}\left[\left(p^{PI}\cdot S(T)\right)^{l}\cdot\mathds{1}_{p^{PI}\cdot S(T)\geq\alpha^{PI}\cdot V_{0}}\right].\nonumber 
\end{align}

The expected value $\mathbb{E}_{\mathbb{P}}\left[\left(p^{PI}\cdot S(T)\right)^{l}\cdot\mathds{1}_{p^{PI}\cdot S(T)\geq\alpha^{PI}\cdot V_{0}}\right]$
is derived within the scope of the calculation of the fair value of
so-called power options.%
\footnote{See, e.g., \citep{HeynenKat1996} or \citep{MacovschiQuittardPinon2006}.%
} It basically consists in a repeated application of the change of
probability measure defined in (\ref{eq:P_k_tilde}). Thus, we obtain
for $l\in\mathbb{N}_{0}$ 
\begin{align}
\mathbb{E}_{\mathbb{P}}\left[\left(p^{PI}\cdot S(T)\right)^{l}\cdot\mathds{1}_{p^{PI}\cdot S(T)\geq\alpha^{PI}\cdot V_{0}}\right] & =\left(p^{PI}\cdot S_{0}\right)^{l}\cdot e^{l\cdot\mu_{S}\cdot T+\frac{1}{2}\cdot l\cdot(l-1)\cdot\sigma_{S}^{2}\cdot T}\cdot\Phi\left(d_{1,l}\right),\label{eq:U_power}\\
\text{where}\ \ d_{1,l} & =\frac{\ln\left(\frac{p^{PI}\cdot S_{0}}{\alpha^{PI}\cdot V_{0}}\right)+\left[\mu_{S}+\left(l-\frac{1}{2}\right)\cdot\sigma_{S}^{2}\right]\cdot T}{\sigma_{S}\cdot\sqrt{T}}.\nonumber 
\end{align}
Hence, substituting (\ref{eq:UPM_power}) and (\ref{eq:U_power})
in (\ref{eq:m_k_OBPI_UPM}) finally leads to 
\begin{align*}
 & m_{k}\left(V^{OBPI}(T)\right)\\
 & =\left(\alpha^{PI}\cdot V_{0}\right)^{k}\\
 & +\sum_{i=1}^{k}{k \choose i}\cdot\left(\alpha^{PI}\cdot V_{0}\right)^{k-i}\cdot\sum_{l=0}^{i}{i \choose l}\cdot(-1)^{i-l}\cdot\left(\alpha^{PI}\cdot V_{0}\right)^{i-l}\cdot\mathbb{E}_{\mathbb{P}}\left[\left(p^{PI}\cdot S(T)\right)^{l}\cdot\mathds{1}_{p^{PI}\cdot S(T)\geq\alpha^{PI}\cdot V_{0}}\right]\\
 & =\left(\alpha^{PI}\cdot V_{0}\right)^{k}\\
 & +\left(\alpha^{PI}\cdot V_{0}\right)^{k}\cdot\sum_{i=1}^{k}\sum_{l=0}^{i}{k \choose i}\cdot{i \choose l}\cdot(-1)^{i-l}\cdot\left(\frac{p^{PI}\cdot S_{0}}{\alpha^{PI}\cdot V_{0}}\right)^{l}\cdot e^{l\cdot\left[\mu_{S}+\frac{1}{2}\cdot(l-1)\cdot\sigma_{S}^{2}\right]\cdot T}\cdot\Phi\left(d_{1,l}\right).
\end{align*}

\bibliographystyle{plainnat}
\bibliography{PhDreferences}

\end{document}